\documentclass[11pt]{article}

\usepackage[a4paper,margin=1in]{geometry}
\usepackage[T1]{fontenc}
\usepackage{lmodern}
\usepackage{amsmath,amssymb,amsthm,mathtools}
\usepackage{enumitem}
\usepackage{needspace}
\usepackage{booktabs,tabularx,array}
\usepackage{placeins}
\usepackage{hyperref}

\numberwithin{equation}{section}

\hypersetup{
  colorlinks=true,
  linkcolor=blue,
  citecolor=blue,
  urlcolor=blue
}

\newtheorem{theorem}{Theorem}[section]
\newtheorem{lemma}[theorem]{Lemma}
\newtheorem{proposition}[theorem]{Proposition}
\newtheorem{corollary}[theorem]{Corollary}
\newtheorem{remark}[theorem]{Remark}

\newtheorem*{definition*}{Definition}

\newcommand{\F}{\mathbb F}
\newcommand{\Fq}{\F_q}
\newcommand{\Fp}{\F_p}
\newcommand{\Fthree}{\F_3}
\newcommand{\Fqstar}{\F_q^*}
\newcommand{\Fthreestar}{\F_3^*}
\newcommand{\Tr}{\operatorname{Tr}}
\newcommand{\DDT}{\operatorname{DDT}}
\newcommand{\BCT}{\operatorname{BCT}}
\newcommand{\CCZ}{\operatorname{CCZ}}
\newcommand{\EA}{\operatorname{EA}}

\newcolumntype{Y}{>{\raggedright\arraybackslash}X}

\title{On APN Functions with Boomerang Uniformity One over $\mathbb F_{3^n}$:\\
Differential and Boomerang Spectra and CCZ-Inequivalence}
\author{ Namhun Koo$^1$, Soonhak Kwon$^{2,3}$, Minwoo Ko$^2$, Byunguk Kim$^2$\\
	\small{\texttt{ Email: komaton@skku.edu, shkwon@skku.edu,}}\\
	\small{\texttt{minwoo1403@skku.edu, kbu0923@g.skku.edu}}\\
	\small{$^1$Institute of Basic Science, Sungkyunkwan University, Suwon, Korea}\\
	\small{$^2$Department of Mathematics, Sungkyunkwan University, Suwon, Korea}\\
	\small{$^3$Applied Algebra and Optimization Research Center, Sungkyunkwan University, Suwon, Korea}
}
\date{\small September 8, 2026}

\begin{document}
\maketitle

\begin{abstract}
Let $q=3^n$, where $n>1$ is odd, and let $g:\Fq\to\Fq$ be a
perfect nonlinear (PN) function represented by a
Dembowski--Ostrom (DO) polynomial.  Put $\tau=g(1)$,
let $\epsilon$ be the indicator of $\Fthree^*$, and, for $c\in\Fq$, define
$\widetilde G_c(x):=g(x+c)+\tau\epsilon(x)$.  We prove that every
$\widetilde G_c$ is APN and has boomerang uniformity either one or two.  More precisely,
\[
  \beta_{\widetilde G_c}=1
  \quad\Longleftrightarrow\quad
  c\in\mathcal C_g
  :=\{c\in\Fq\setminus\Fthree:g(c)+\tau\notin g(\Fq)\},
  \qquad
  |\mathcal C_g|=\frac{q-3}{2},
\]
whereas $\beta_{\widetilde G_c}=2$ for the remaining $(q+3)/2$ parameters.
We determine the common differential spectrum and complete boomerang spectra
of all the functions $\widetilde G_c$.
Since boomerang uniformity one is the least possible for an APN function over
a finite field of odd characteristic, this gives, to the best of our
knowledge, the first general construction yielding infinite families of
APN functions attaining this optimum.

This common differential spectrum rules out CCZ equivalence with every
power function and every Ness--Helleseth-type binomial.  We also prove that
CCZ equivalence between sign-switches of DO PN functions forces EA equivalence
between the original PN functions.  Using the orders of the nuclei of the
associated presemifields, we exhibit, for
infinitely many odd $n$, three pairwise CCZ-inequivalent PN functions over
$\F_{3^n}$, one from each of the Gold $f_1$, Ding--Yuan $f_3$, and Bierbrauer
$f_5$ families.
Consequently, over each such field, our construction produces three pairwise
CCZ-inequivalent APN functions with boomerang uniformity one.  The smallest
extension degree obtained in this way is $n=45$.

	
\end{abstract}

\section{Introduction}

Let $q$ be an odd prime power of characteristic $p$ and let $\Fq$ be the
finite field with $q$ elements.  Differential and boomerang properties of functions on $\Fq$ are
fundamental criteria in the study of cryptographic mappings.  For a function
$F:\Fq\to\Fq$ and $a\in\Fq$, write
\[
  D_aF(x):=F(x+a)-F(x),
\]
and define
\[
  L_{a,F}(x)
  :=D_aF(x)-D_aF(0)
  =F(x+a)-F(x)-F(a)+F(0),
\]
the nonconstant part of the derivative $D_aF$.  When the underlying function
is clear, we simply write $L_a$.  For a general function $F$, the map $L_{a,F}$
need not be linear.  If $F$ is quadratic, then $D_aF$ is affine and $L_{a,F}$
is its $\Fp$-linear part; in this case we call $L_{a,F}$ the
\emph{linearized derivative} of $F$ in direction $a$.

For a function $F:\Fq\to\Fq$ and $a,b\in\Fq$, define the entry of the
difference distribution table (DDT) of $F$ at $(a,b)$ by
\[
  \delta_F(a,b):=\DDT_F(a,b):=
  \#\{x\in\Fq:D_aF(x)=b\}.
\]
Following Nyberg~\cite{Nyberg94}, the differential uniformity of $F$ is
\[
  \delta_F:=\max_{a\in\Fq^*,\ b\in\Fq}\delta_F(a,b).
\]
The function $F$ is perfect nonlinear (PN), also called planar in odd
characteristic, when $\delta_F=1$, and it is almost perfect nonlinear (APN)
when $\delta_F=2$.

The boomerang connectivity table was introduced for permutations by Cid
et al.~\cite{CidEtAl18} and was extended to arbitrary functions by Li
et al.~\cite{LiQuSunLi19}.  For $a,b\in\Fq^*$, its entries and the boomerang
uniformity are
\[
  \beta_F(a,b):=\BCT_F(a,b):=
  \#\{(x,y)\in\Fq^2:
       F(y+a)-F(x+a)=b,\ F(y)-F(x)=b\}
\]
and
\[
  \beta_F:=\max_{a,b\in\Fq^*}\beta_F(a,b),
\]
respectively.

The very small values $\beta_F=0,1,2$ have recently attracted particular
attention over finite fields of odd characteristic.  Every PN function has boomerang uniformity zero, and
recent work has also produced non-PN families with boomerang uniformity zero.
A particularly active line of research concerns the Ness--Helleseth-type binomials
\[
  F_{r,u}(x)=x^r(1+u\chi(x)),
\]
where $\chi$ is the quadratic character of $\Fq$, extended by $\chi(0)=0$.
Lyu, Wang, and Zheng determined the differential and boomerang spectra of
$F_{q-2,\pm1}$ and obtained classes with boomerang uniformity zero or at most one
\cite{LyuWangZheng24}.  Koo and Kwon studied $F_{(q+1)/4,\pm1}$ and found
classes with boomerang uniformity zero, one, or two \cite{KK26}.  Further
characteristic-$3$ classes with boomerang uniformity zero or one were obtained
by Koo et al.~\cite{KKKB26b,KKKB26}, while other locally-APN
classes with boomerang uniformity two were studied in
\cite{HuEtAl23,MW25}.

In the terminology of \cite[Definition~2.1]{KK26}, the terms locally-PN and
locally-APN, used for monomials as well as for the Ness--Helleseth-type
binomials above, refer, respectively, to the conditions
\[
  \delta_F(1,b)\le1
  \quad\text{and}\quad
  \delta_F(1,b)\le2
  \qquad(b\in\Fq\setminus\Fp).
\]
Our exhaustive BCT computations for the small-field APN binomials recorded by
Budaghyan and Pal \cite[Table~5]{BudaghyanPal25} show that some have boomerang
uniformity one; one such example is discussed in
Subsection~\ref{sec:apn-beta-lower-bound}.  Their conjecture that the binomial
class contains an infinite APN subfamily was subsequently disproved in
\cite{BartoliStanica26,MW25}.  Thus these examples did not provide a general
construction of infinite families of APN functions with boomerang
uniformity one.

For boomerang uniformity two, Pal and St\u{a}nic\u{a} related the BCT entries of
odd APN functions $F$, that is, APN functions satisfying $F(-x)=-F(x)$,
to the corresponding $(-1)$-DDT entries and determined
conditions under which the inverse function is APN with boomerang uniformity two
\cite{PalStanica25}.  This gives an infinite family.  To the best of our
knowledge, prior to the present work, the inverse family appears to have been
the only established infinite family of APN functions over finite
fields of odd characteristic with boomerang uniformity two.

Apart from PN functions, which have boomerang uniformity zero, all previously
known infinite non-PN families with boomerang uniformity at most one were
locally-PN or locally-APN families; none of the families with boomerang
uniformity one yielded an infinite family of APN functions.  To the best of
our knowledge, APN functions with boomerang uniformity one were known only
through isolated examples over small fields.  The present paper fills this
gap by giving a general construction that produces several such infinite
families.

Throughout the paper, $q=3^n$ with $n>1$ odd.  Thus $q\equiv3\pmod4$.

Let $g:\Fq\to\Fq$ be represented by a Dembowski--Ostrom (DO)
polynomial~\cite{DO68}, i.e.\ a polynomial of the form
\[
  g(x)=\sum_{0\le i\le j<n} c_{ij}x^{3^i+3^j}
\]
up to reduction modulo $x^q-x$.  We assume that $g$ is perfect nonlinear (PN), also called planar in odd characteristic, meaning that
the derivative $D_ag$ is a permutation of $\Fq$ for every $a\in\Fq^*$.

Unless the underlying function is explicitly specified otherwise, $L_a$
denotes the linearized derivative of this fixed function $g$.  Thus
\[
  L_a(x):=L_{a,g}(x)=D_ag(x)-D_ag(0)
  =g(x+a)-g(x)-g(a).
\]
Put
\begin{equation}\label{eq:tau}
  \tau:=g(1).
\end{equation}
We shall see below that $\tau\ne0$.  Define the sign-switch $G$ of $g$ by
\[
  G(x)=
  \begin{cases}
    -g(x), & x\in\Fthree,\\
     g(x), & x\in\Fq\setminus\Fthree.
  \end{cases}
\]
Since $g(0)=0$ and $g(\pm1)=\tau$, equivalently,
\begin{equation}\label{eq:sign-switch}
  G(x):=g(x)+\tau\epsilon(x),
  \qquad
  \epsilon(x):=
  \begin{cases}
    1, & x=\pm1,\\
    0, & \text{otherwise}.
  \end{cases}
\end{equation}

Earlier modification and switching constructions from PN functions over fields
of odd characteristic yield, up to EA equivalence, certain ternary Gold
instances of the sign-switch $G$ above, thereby already establishing their APN
property~\cite{XuCaoXu16,ZhaHu13}.  Apart from computational differential
spectra for a few small examples in~\cite{XuCaoXu16}, those works do not give
general exact differential and boomerang spectra or study CCZ equivalence among
sign-switches arising from distinct PN inputs.

For a map $F:\Fq\to\Fq$ and $i\ge0$, define the DDT and BCT
entry counts by
\[
\begin{aligned}
  N_i^{\DDT}(F)
    &:=\#\{(a,b)\in\Fq^*\times\Fq:\delta_F(a,b)=i\},\\
  N_i^{\BCT}(F)
    &:=\#\{(a,b)\in\Fq^*\times\Fq^*:\beta_F(a,b)=i\}.
\end{aligned}
\]
Using $v^m$ for $m$ copies of $v$, with $v^0$ omitted, define the differential
and boomerang spectra of $F$ by the multisets
\begin{equation}\label{eq:global-spectra}
\begin{aligned}
  \mathcal S_{\DDT}(F)
    &:=\{\!\{\,0^{N_0^{\DDT}(F)},\ldots,
         \delta_F^{N_{\delta_F}^{\DDT}(F)}\,\}\!\},\\
  \mathcal S_{\BCT}(F)
    &:=\{\!\{\,0^{N_0^{\BCT}(F)},\ldots,
         \beta_F^{N_{\beta_F}^{\BCT}(F)}\,\}\!\}.
\end{aligned}
\end{equation}

\subsection*{Main contributions of the paper}

For the fixed DO PN function $g:\Fq\to\Fq$, with $\tau$
and $G$ as in \eqref{eq:tau}--\eqref{eq:sign-switch}, define
\begin{equation}\label{eq:main-family}
\begin{aligned}
  \mathcal C_g
    &:=\{c\in\Fq\setminus\Fthree:g(c)+\tau\notin g(\Fq)\},\\
  \widetilde G_c(x)
    &:=g(x+c)+\tau\epsilon(x)
      =G(x)+D_cg(x)
      \qquad(c\in\Fq).
\end{aligned}
\end{equation}
Our main construction, given in
Corollary~\ref{cor:natural-derivative-perturbations}, yields
\[
  |\mathcal C_g|=\frac{q-3}{2},
  \qquad
  \delta_{\widetilde G_c}=2\quad(c\in\Fq),
  \qquad
  \beta_{\widetilde G_c}=
  \begin{cases}
    1,&c\in\mathcal C_g,\\
    2,&c\in\Fq\setminus\mathcal C_g.
  \end{cases}
\]
Thus every DO PN function gives exactly $(q-3)/2$ parameter values for which
the corresponding natural perturbation is APN and has boomerang uniformity one.
Since $\widetilde G_c\sim_{\EA}G$, all members of the family share the exact
differential spectrum of $G$ determined in
Theorem~\ref{thm:diff-spectrum}.  Corollary~\ref{cor:lc-bct-spectrum}
determines the complete BCT row distributions and hence the boomerang spectra
for every $c\in\Fq$: it lists three possible types when
$\beta_{\widetilde G_c}=1$ and five when $\beta_{\widetilde G_c}=2$.

For functions with $\delta_F\le2$ over finite fields of odd characteristic,
Proposition~\ref{prop:beta-zero-pn} identifies boomerang uniformity zero with
perfect nonlinearity.  Hence
Corollary~\ref{cor:apn-beta-lower-bound} shows that boomerang uniformity one
is optimal for APN functions.

The construction begins with the sign-switch $G$.
Proposition~\ref{prop:zero-sum}, together with the two-to-one property in
Theorem~\ref{thm:CM-two-to-one}, yields the disjoint-union decomposition
\[
  \Fq^*=g(\Fq^*)\sqcup\bigl(-g(\Fq^*)\bigr).
\]
Thus $g(\Fq^*)$ contains exactly one element from each pair $\{a,-a\}$ in
$\Fq^*$.  Corollary~\ref{cor:special-preimages} gives the resulting special
preimages $g^{-1}(0)=\{0\}$, $g^{-1}(\tau)=\{\pm1\}$, and
$g^{-1}(-\tau)=\varnothing$, which are used throughout the paper.
Theorems~\ref{thm:diff-spectrum} and~\ref{thm:bct-spectrum} determine the
exact differential and boomerang spectra of $G$; in particular,
\[
  \delta_G=\beta_G=2.
\]

Every $\Fthree$-affine perturbation $G_A=G+A$ is APN and has
$\beta_{G_A}\in\{1,2\}$.  Theorem~\ref{thm:affine-criterion-spectrum} gives a
necessary and sufficient condition for $\beta_{G_A}=1$ and, in that case,
determines the exact BCT row distribution and boomerang spectrum.
Proposition~\ref{prop:admissible-lc} specializes this criterion to the linear
part $M=L_c$; adding the constant $g(c)$ gives the derivative
$D_cg=L_c+g(c)$ and hence the natural family \eqref{eq:main-family}.

We also compare the differential and boomerang spectra of $G$ with those of
power functions and Ness--Helleseth-type binomials
$x^r(1+u\chi(x))$ with $u\ne0$.  For each such function, all DDT rows indexed
by $a\ne0$ have the same value distribution, and likewise all BCT rows indexed by $a\ne0$ have the
same value distribution.  The resulting DDT and BCT entry counts are therefore
divisible by $q-1$.  In contrast,
Theorems~\ref{thm:diff-spectrum} and~\ref{thm:bct-spectrum} give
\[
  N_0^{\DDT}(G)=N_2^{\BCT}(G)=4(q-3).
\]
Theorem~\ref{thm:spectral-separation} consequently shows
that both spectra of $G$ differ from those of every function in these
two classes.  Since the differential spectrum is CCZ-invariant, the DDT
comparison yields CCZ-inequivalence.  The BCT comparison gives an additional
spectral distinction, but it is not used for
CCZ-inequivalence because the boomerang spectrum is not a CCZ invariant in
general.

The construction applies to every DO PN function over
$\F_{3^n}$ with $n>1$ odd.  Subsection~\ref{sec:known-examples} applies the
main construction to the Gold $f_1$, Ding--Yuan $f_3$,
Zha--Kyureghyan--Wang $f_4$, and Bierbrauer $f_5$ families.  Each valid family
member gives natural APN functions with boomerang uniformity one for exactly
$(q-3)/2$ parameters, together with their complete differential and
boomerang spectra.  For $g(x)=x^{10}$,
Remark~\ref{rem:gold-all-kappa} proves that all eight boomerang spectrum types
in Corollary~\ref{cor:lc-bct-spectrum} occur for every odd $n\ge9$.
Table~\ref{tab:kappa-evidence} reports exhaustive computations realizing the
six types in Corollary~\ref{cor:lc-bct-spectrum}\textup{(a)} and \textup{(b)}
for selected members of the $f_1$, $f_3$, and $f_5$ families.  Together with
the two special types in parts \textup{(c)} and \textup{(d)}, all eight types
occur for each selected function.

Finally, Theorem~\ref{thm:switch-rigidity} proves that sign-switching does not
collapse CCZ classes of the underlying DO PN functions.  More precisely, if
$G$ and $H$ are the sign-switches of $g$ and $h$, respectively, then
\[
  G\sim_{\CCZ}H
  \quad\Longrightarrow\quad
  g\sim_{\EA}h.
\]
We refer to this implication as \emph{switch-rigidity}.  Its proof combines an
intrinsic characterization of the vertical subspace $\{0\}\times\Fq$, a
CCZ-to-EA collapse for switched maps, and a theorem establishing the
uniqueness of the underlying PN function.

Since every affine perturbation constructed above is EA-equivalent to its
sign-switch,
Corollary~\ref{cor:ccz-inequivalent-associated-functions} transfers pairwise
CCZ-inequivalence of the underlying DO PN functions to associated APN
functions with boomerang uniformity one.  Using the orders of the nuclei of
the associated presemifields,
Theorem~\ref{thm:three-beta-one-classes} obtains three such functions from one
member of each of the Gold $f_1$, Ding--Yuan $f_3$, and Bierbrauer $f_5$
families.  The specialization $s=5\ell$, $t=\ell$, with odd $\ell>1$, gives
$n=15\ell$, beginning with $n=45$.

The remainder of the paper is organized as follows.
Section~\ref{sec:preliminaries} establishes the required structural facts
about DO PN functions;
Section~\ref{sec:spectra} determines the spectra of $G$, proves the comparison
results, and establishes optimality;
Section~\ref{sec:affine-perturbations} presents the affine-perturbation
construction and its applications; and
Section~\ref{sec:switch-rigidity-inequivalence} proves switch-rigidity and the
resulting CCZ-inequivalence results.
Section~\ref{sec:conclusion} concludes the paper and records directions for
further work.

\section{Preliminaries for DO PN functions}\label{sec:preliminaries}

We first record the general facts for DO PN functions that will be used
throughout the paper.  The statements through Proposition~\ref{prop:zero-sum} are
formulated for arbitrary prime powers $q=p^m$ with $q\equiv3\pmod4$, as required for
the application of Feng and Luo~\cite[Lemma~3(ii)]{FL07} below.  In this general part,
a DO polynomial over $\F_q$ means a polynomial whose monomials have exponents
$p^i+p^j$, that is,
\[
  f(x)=\sum_{0\le i\le j<m} c_{ij}x^{p^i+p^j}
\]
up to reduction modulo $x^q-x$.  From Corollary~\ref{cor:special-preimages} onward, we
return to the standing assumption $q=3^n$ with $n$ odd.

For a map $f:\Fq\to\Fq$ with $f(0)=0$, we say that $f$ is \emph{two-to-one} if
\[
  f^{-1}(0)=\{0\},
\]
and every nonzero value in $f(\Fq)$ has exactly two preimages.

\begin{theorem}[Coulter--Matthews {\cite[Theorem~3]{CM11}}]\label{thm:CM-two-to-one}
Let $q$ be an odd prime power, and let $f\in\Fq[x]$ be a DO polynomial.
Then $f$ is planar over $\Fq$ if and only if $f$ is two-to-one.  Equivalently, $f$ is planar over $\Fq$ if and only if
\[
  |f(\Fq)|=\frac{q+1}{2}.
\]
\end{theorem}

For $t\in\Fp$, every DO polynomial $f$ satisfies $f(tx)=t^2f(x)$; in particular,
$f(0)=0$ and $f(-x)=f(x)$.  Hence
Theorem~\ref{thm:CM-two-to-one} has the following consequence: if $f$ is a
DO PN function, then
\[
  f^{-1}(0)=\{0\},
\]
and for every nonzero value $c\in f(\Fq)$ one has
\[
  \#f^{-1}(c)=2.
\]
Indeed, the domain is partitioned into the $(q+1)/2$ orbits
\[
  \{0\},\qquad \{x,-x\}\quad(x\in\Fq^*),
\]
and Theorem~\ref{thm:CM-two-to-one} says that these orbits give exactly $(q+1)/2$ distinct
values.

The following lemma is an immediate consequence of
Feng and Luo~\cite[Lemma~3(ii)]{FL07}.

\begin{lemma}\label{lem:FL-walsh}
Let $q=p^m$ be a prime power with $q\equiv3\pmod4$, and let $f$ be a
DO PN function on $\Fq$.  For $a,b\in\Fq$, define the Walsh transform of $f$ by
\[
  W_f(a,b):=\sum_{x\in\Fq}\zeta^{\Tr(bf(x)-ax)},
\]
where $\zeta=e^{2\pi i/p}$ is a primitive $p$-th root of unity and
$\Tr:\Fq\to\Fp$ is the absolute trace map.  Then, for every $b\in\Fq^*$,
\[
  W_f(0,b)^2=-q.
\]
\end{lemma}

\begin{proposition}\label{prop:zero-sum}
Let $q=p^m$ be a prime power with $q\equiv3\pmod4$, and let $f$ be a
DO PN function on $\Fq$.  Then
\[
  f(x)+f(y)=0
  \qquad\Longrightarrow\qquad
  x=y=0.
\]
\end{proposition}

\begin{proof}
Let
\[
  S:=\{(x,y)\in\Fq\times\Fq:f(x)+f(y)=0\}.
\]
Clearly $(0,0)\in S$.  For $b\in\Fq$,
\[
  W_f(0,b)^2
   =\sum_{x\in\Fq}\zeta^{\Tr(bf(x))}
      \sum_{y\in\Fq}\zeta^{\Tr(bf(y))}
   =\sum_{x,y\in\Fq}\zeta^{\Tr(b(f(x)+f(y)))}.
\]
Summing over $b\in\Fq$ and using the orthogonality of additive characters gives
\[
\begin{aligned}
  \sum_{b\in\Fq}W_f(0,b)^2
   &=\sum_{(x,y)\in S}\sum_{b\in\Fq}1
     +\sum_{(x,y)\notin S}\sum_{b\in\Fq}\zeta^{\Tr(b(f(x)+f(y)))}  \\
   &=|S|q.
\end{aligned}
\]
On the other hand, by Lemma~\ref{lem:FL-walsh},
\[
\begin{aligned}
  \sum_{b\in\Fq}W_f(0,b)^2
   &=W_f(0,0)^2+\sum_{b\in\Fq^*}W_f(0,b)^2 \\
   &=q^2+(q-1)(-q)=q.
\end{aligned}
\]
Thus $|S|=1$.  Since $(0,0)\in S$, we get $S=\{(0,0)\}$.

\end{proof}

We next combine Theorem~\ref{thm:CM-two-to-one} and
Proposition~\ref{prop:zero-sum}.  The two-to-one property gives
$f^{-1}(0)=\{0\}$ and
\[
  \bigl|f(\Fq^*)\bigr|=\frac{q-1}{2},
\]
whereas Proposition~\ref{prop:zero-sum} gives
\[
  f(\Fq^*)\cap\bigl(-f(\Fq^*)\bigr)=\varnothing.
\]
Since these two subsets of $\Fq^*$ are disjoint and each has $(q-1)/2$
elements, they give the disjoint-union decomposition
\begin{equation}\label{eq:nonzero-image-disjoint-union}
  \Fq^*=f(\Fq^*)\sqcup\bigl(-f(\Fq^*)\bigr).
\end{equation}
Equivalently, the partition \eqref{eq:nonzero-image-disjoint-union} says that
$f(\Fq^*)$ contains exactly one element from each pair $\{a,-a\}$ in
$\Fq^*$.

We now apply the disjoint-union decomposition
\eqref{eq:nonzero-image-disjoint-union} with $q=3^n$ and $f=g$.

\begin{corollary}\label{cor:special-preimages}
Let $q=3^n$ with $n$ odd, and let $g$ be a DO PN function on $\Fq$.  Put
$\tau=g(1)$.  Then $\tau\ne0$, and
\[
  g^{-1}(0)=\{0\},\qquad
  g^{-1}(\tau)=\{\pm1\},\qquad
  g^{-1}(-\tau)=\varnothing.
\]
Consequently, $g$ attains none of $0,\pm\tau$ on
$\Fq\setminus\Fthree$:
\[
  g(\Fq\setminus\Fthree)\cap \tau\Fthree=\varnothing,
  \qquad
  \tau\Fthree:=\{0,\tau,-\tau\}.
\]
\end{corollary}

\begin{proof}
Since $3^n\equiv3\pmod4$ for odd $n$, the disjoint-union decomposition
\eqref{eq:nonzero-image-disjoint-union} applies to $g$.
Proposition~\ref{prop:zero-sum}, with one input equal to zero, gives
$g^{-1}(0)=\{0\}$, so $\tau=g(1)\ne0$.  Also
$g(1)=g(-1)=\tau$, and the two-to-one property gives
\[
  g^{-1}(\tau)=\{\pm1\}.
\]
Since $\tau\in g(\Fq^*)$ and
$g(\Fq^*)\cap(-g(\Fq^*))=\varnothing$, one has
$-\tau\notin g(\Fq^*)$.  Hence $g^{-1}(-\tau)=\varnothing$.

It remains to prove
$g(\Fq\setminus\Fthree)\cap\tau\Fthree=\varnothing$.  If
$x\notin\Fthree$ and $g(x)\in\tau\Fthree$, then $g(x)$ is one of
$0,\tau,-\tau$, contradicting the three preimage statements just proved.
\end{proof}

The preimage statements in Corollary~\ref{cor:special-preimages} will be used
in the proofs of the exact differential spectrum
(Theorem~\ref{thm:diff-spectrum}), the exact boomerang spectrum
(Theorem~\ref{thm:bct-spectrum}), and the affine criterion
(Theorem~\ref{thm:affine-criterion-spectrum}).

Recall that $L_a$ is the linearized derivative of the fixed
DO polynomial $g$.  Since $g(0)=0$, one has
\[
  D_ag(x)=L_a(x)+g(a).
\]
Moreover, the map $(a,x)\mapsto L_a(x)$ is symmetric and $\Fthree$-bilinear.  In particular,
\[
  L_a(x)=L_x(a),
\]
and, for each fixed $a$, $L_a$ is $\Fthree$-linear in $x$.

\begin{lemma}\label{lem:L1-image}
One has
\[
  L_1(\Fthree)=\tau\Fthree.
\]
\end{lemma}

\begin{proof}
For $t\in\Fthree$, the homogeneity of $g$ gives
\[
  L_1(t)=g(1+t)-g(1)-g(t)
  =\bigl((1+t)^2-1-t^2\bigr)\tau=2t\tau=-t\tau.
\]
Hence $L_1(\Fthree)=\tau\Fthree$.
\end{proof}

With $\epsilon$ and $G$ as in \eqref{eq:sign-switch}, we shall use the
following directional-derivative notation.  For $a,x\in\Fq$, put
\begin{equation}\label{eq:switching-derivative}
\begin{aligned}
  \epsilon_a(x)&:=D_a\epsilon(x)=\epsilon(x+a)-\epsilon(x),\\
  D_aG(x)&=D_ag(x)+\tau\epsilon_a(x)
          =L_a(x)+g(a)+\tau\epsilon_a(x).
\end{aligned}
\end{equation}

\begin{lemma}\label{lem:quadratic-identities}
For the fixed DO polynomial $g$, the following identities hold
for all $x,y\in\Fq$:
\[
  g(x+y)+g(x-y)=-g(x)-g(y),
\]
and
\[
  g(x+y)-g(x-y)=-L_y(x)=-L_x(y).
\]
\end{lemma}

\begin{proof}
By the definition of $L$, the evenness $g(-y)=g(y)$, and the $\Fthree$-linearity of
$L$ in the subscript,
\[
  g(x+y)=g(x)+g(y)+L_y(x),
  \qquad
  g(x-y)=g(x)+g(y)-L_y(x).
\]
Adding and subtracting these two identities gives the result.  The last
equality follows from the symmetry $L_y(x)=L_x(y)$.
\end{proof}

\subsection*{Equivalence conventions and the differential spectrum}

For later use, we briefly recall some standard facts about CCZ equivalence
and the CCZ-invariance of the differential spectrum.

For a function $F:\Fq\to\Fq$, write
\[
  \mathcal G_F:=\{(x,F(x)):x\in\Fq\}\subseteq\Fq\times\Fq
\]
for its graph.  For $v=(a,b)\in\Fq\times\Fq$, we also write
\[
  \delta_F(v):=\delta_F(a,b).
\]
With this notation, the DDT entry $\delta_F(v)$ has the graph-intersection form
\[
  \delta_F(v)
  =|\mathcal G_F\cap(\mathcal G_F-v)|
  =|\mathcal G_F\cap(\mathcal G_F+v)|.
\]
Thus $\delta_F(v)=0$ if and only if the two sets $\mathcal G_F$ and
$\mathcal G_F-v$ are disjoint.

Two functions $F,F':\Fq\to\Fq$ are CCZ-equivalent in the sense of Carlet,
Charpin, and Zinoviev~\cite{CCZ98}, written $F\sim_{\CCZ}F'$, if there exists
an affine automorphism $\mathcal A(P)=\mathcal M(P)+\mathbf c$ of the
$\Fthree$-vector space $\Fq\times\Fq$ such that
$\mathcal A(\mathcal G_F)=\mathcal G_{F'}$.
They are EA-equivalent, written $F\sim_{\EA}F'$, if there exist invertible
$\Fthree$-linear maps $M_1,M_2:\Fq\to\Fq$, an $\Fthree$-linear map
$B:\Fq\to\Fq$, and constants $a_0,b_0\in\Fq$ such that
\[
  F'(M_1x+a_0)=M_2F(x)+B(x)+b_0
  \qquad(x\in\Fq).
\]

\begin{lemma}\label{lem:ddt-entry-transport}
Suppose $\mathcal A(P)=\mathcal M(P)+\mathbf c$ sends $\mathcal G_F$ to
$\mathcal G_{F'}$. Then
\[
  \delta_F(v)=\delta_{F'}(\mathcal Mv)
  \qquad(v\in\Fq\times\Fq).
\]
In particular, if $\delta_F(v)=0$, then $\delta_{F'}(\mathcal Mv)=0$.
\end{lemma}

\begin{proof}
For $v\in\Fq\times\Fq$, the affine bijection $\mathcal A$ restricts to a bijection
\[
  \mathcal A:\ \mathcal G_F\cap(\mathcal G_F-v)
  \xrightarrow{\sim}
  \mathcal G_{F'}\cap(\mathcal G_{F'}-\mathcal Mv).
\]
Taking cardinalities gives the claim.
\end{proof}

\begin{lemma}
\label{lem:ccz-global-ddt-spectrum}
If $F,F':\Fq\to\Fq$ are CCZ-equivalent, then
\[
  \mathcal S_{\DDT}(F)=\mathcal S_{\DDT}(F').
\]
\end{lemma}

\begin{proof}
By Lemma~\ref{lem:ddt-entry-transport}, the multiset of
$\delta_F(v)$ over all nonzero $v\in\Fq\times\Fq$ is CCZ-invariant.
Since $\delta_F(0,b)=0$ for every $b\in\Fq^*$, this multiset is obtained
from $\mathcal S_{\DDT}(F)$ by adjoining $q-1$ zeros, and the claim follows.
\end{proof}

\section{Exact differential and boomerang spectra and comparison}\label{sec:spectra}

For a map $F:\Fq\to\Fq$, a fixed $a\in\Fq^*$, and $b\in\Fq$, the
\emph{fiber of $D_aF$ over $b$} is
\[
  (D_aF)^{-1}(b)=\{x\in\Fq:D_aF(x)=b\}.
\]
In what follows, the unqualified term \emph{fiber} always means a fiber of a
derivative.

\subsection{The exact differential spectrum}\label{sec:differential}

We now determine the differential rows of $G$.  We begin with a small symmetry lemma which
will be used repeatedly.  Recall that a map $f:\Fq\to\Fq$ is called \emph{even} if
$f(-x)=f(x)$ for all $x\in\Fq$.

\begin{lemma}\label{lem:even-derivative}
Let $f:\Fq\to\Fq$ be even.  For $a\in\Fq$, let
$\iota_a:\Fq\to\Fq$ be the involution defined by
\[
  \iota_a(x):=-x-a.
\]
Then, for every $x\in\Fq$,
\[
  D_{-a}f(x)=D_af(-x),
  \qquad
  D_af(\iota_a(x))=-D_af(x).
\]
Consequently, for every $b\in\Fq$, the involution $\iota_a$ maps the
fiber of $D_af$ over $b$ bijectively onto the fiber over $-b$.
Moreover, if $D_af(x)=D_af(y)$, then
\[
  f(\iota_a(y))-f(\iota_a(x))=f(y)-f(x).
\]
\end{lemma}

\begin{proof}
By the evenness of $f$,
\begin{align*}
  D_{-a}f(x)
   &=f(x-a)-f(x)
     =f(-x+a)-f(-x)
     =D_af(-x),\\
  D_af(\iota_a(x))
   &=f(-x)-f(-x-a)
     =f(x)-f(x+a)
     =-D_af(x).
\end{align*}
The fiber assertion follows from the second identity and the fact that
$\iota_a$ is an involution.  Finally, if $D_af(x)=D_af(y)$, then, by the
evenness of $f$,
\[
  f(\iota_a(y))-f(\iota_a(x))
  =f(y+a)-f(x+a)
  =f(y)-f(x).
\]
Here the last equality follows from $D_af(x)=D_af(y)$.
\end{proof}

The functions $g$, $G$, and $\epsilon$ in~\eqref{eq:sign-switch} are even.  Hence Lemma~\ref{lem:even-derivative}
applies to all three of them.

For each $a\in\Fq^*$, define the \emph{exceptional input set} for the
sign-switch $g\mapsto G$ by
\[
  E_a:=\{x\in\Fq:D_aG(x)\ne D_ag(x)\}
      =\{x\in\Fq:\epsilon_a(x)\ne0\}.
\]
The equality follows from \eqref{eq:switching-derivative}, since $\tau\ne0$.
Thus $E_a$ is a subset of the domain.  Its images $D_ag(E_a)$ and
$D_aG(E_a)$, which are subsets of the codomain, will be called the
\emph{old} and \emph{new exceptional derivative-value sets}, respectively.
Here ``old'' and ``new'' always refer to the derivatives before and after the
sign-switch $g\mapsto G$.  By definition, $D_aG=D_ag$ on
$\Fq\setminus E_a$.

For a DDT or BCT row, we use the same exponent notation for its value
distribution: the row has type $0^{m_0}1^{m_1}\cdots$ if it contains exactly
$m_i$ entries equal to $i$; terms with $m_i=0$ are omitted.

\begin{theorem}[Exact differential spectrum]\label{thm:diff-spectrum}
Let $q=3^n$ with $n>1$ odd, let $g:\Fq\to\Fq$ be a
DO PN function, put $\tau=g(1)$, and let $G$ be its
sign-switch:
\[
  G(x)=
  \begin{cases}
    -g(x), & x\in\Fthree,\\
     g(x), & x\in\Fq\setminus\Fthree.
  \end{cases}
\]
Then the following hold.
\begin{enumerate}[label=\textup{(\roman*)}]
\item If $a\in\Fthree^*$, then $D_aG$ is a permutation of $\Fq$.
Equivalently, the $a$-row of the DDT has type
\[
  1^q.
\]
\item If $a\in\Fq^*\setminus\Fthree$, then the $a$-row of the DDT has type
\[
  0^4 1^{q-8}2^4.
\]
\end{enumerate}
Thus, among the DDT rows indexed by $a\in\Fq^*$, two have type $1^q$ and
$q-3$ have type $0^4 1^{q-8}2^4$.
Consequently, $G$ is APN and its differential spectrum is
\[
  \mathcal S_{\DDT}(G)
  =
  \{\!\{\,
    0^{4(q-3)},\,
    1^{\,2q+(q-3)(q-8)},\,
    2^{4(q-3)}
  \,\}\!\}.
\]
\end{theorem}

\begin{proof}
By \eqref{eq:switching-derivative}, $E_a$ consists exactly of the points
$x$ for which exactly one of $x$ and $x+a$ belongs to $\{\pm1\}$.

\medskip
\noindent\textbf{Case 1: $a\in\Fthree^*$.}
It is enough to treat $a=1$, because Lemma~\ref{lem:even-derivative}, applied to
the even map $G$, gives
\[
  D_{-1}G(x)=D_1G(-x).
\]
Thus $D_{-1}G$ is a permutation exactly when $D_1G$ is a permutation.

For $a=1$, the exceptional input set is $E_1=\{0,2\}$, and
\[
  \epsilon_1(0)=1,
  \qquad
  \epsilon_1(2)=2=-1.
\]
Moreover,
\[
  D_1g(0)=g(1)-g(0)=\tau,
  \qquad
  D_1g(2)=g(0)-g(2)=-\tau.
\]
Hence
\[
  D_1G(0)=-\tau,
  \qquad
  D_1G(2)=\tau.
\]
Thus $D_1g(E_1)=D_1G(E_1)=\{\pm\tau\}$: the switch merely exchanges
the two exceptional derivative values.  Therefore $D_1G$ is a permutation,
and consequently $D_2G=D_{-1}G$ is also a permutation.  This proves the
permutation-row statement for $a\in\Fthree^*$.

\medskip
\noindent\textbf{Case 2: $a\in\Fq^*\setminus\Fthree$.}
Here the exceptional input set is
\begin{equation}\label{eq:exceptional-input-set}
  E_a=\{1,2,1-a,2-a\},
\end{equation}
and its four elements are distinct.  At the two inputs $1$ and $2=-1$ we have
\[
  \epsilon_a(1)=\epsilon_a(2)=2=-1.
\]
The remaining two exceptional inputs are obtained from $1$ and $2$ by the involution
$\iota_a(x)=-x-a$:
\[
  \iota_a(2)=1-a,
  \qquad
  \iota_a(1)=2-a.
\]
Thus Lemma~\ref{lem:even-derivative}, applied to $\epsilon$, also gives
\[
  \epsilon_a(1-a)=\epsilon_a(2-a)=1.
\]

Put
\[
  A:=D_ag(1)=g(a+1)-\tau,
  \qquad
  B:=D_ag(2)=g(a-1)-\tau.
\]
By Lemma~\ref{lem:even-derivative}, applied to $g$,
\[
  D_ag(2-a)=-D_ag(1)=-A,
  \qquad
  D_ag(1-a)=-D_ag(2)=-B.
\]
The values of $D_ag$ and $D_aG$ on the exceptional input set are summarized by
\[
\begin{array}{c|cccc}
 x&1&2&1-a&2-a\\ \hline
 D_ag(x)&A&B&-B&-A\\
 D_aG(x)&A-\tau&B-\tau&\tau-B&\tau-A.
\end{array}
\]
Consequently,
\begin{equation}\label{eq:exceptional-derivative-value-sets}
  D_ag(E_a)=\{\pm A,\pm B\},
  \qquad
  D_aG(E_a)=\{\pm(A-\tau),\pm(B-\tau)\}.
\end{equation}
These are, respectively, the old and new exceptional derivative-value sets.

We claim that $D_aG(E_a)$ has four elements and
$D_aG(E_a)\cap D_ag(E_a)=\varnothing$.  It is enough to show
\begin{equation}\label{eq:four-exclusions-general}
  A,
  \ B,
  \ A+B,
  \ A-B
  \notin\tau\Fthree.
\end{equation}
Indeed, every equality among two elements of the new set, or between an element
of the new set and an element of the old set, forces
one of $A,B,A+B,A-B$ to lie in the line $\tau\Fthree$; for example,
\[
  A-\tau=B \Rightarrow A-B=\tau,
  \qquad
  A-\tau=-B \Rightarrow A+B=\tau,
\]
and
\[
  A-\tau=\tau-A \Rightarrow A=\tau.
\]
The remaining equalities are the same elementary check.

We now prove \eqref{eq:four-exclusions-general}.  Since $a\pm1\notin\Fthree$, Corollary~\ref{cor:special-preimages} gives
\[
  g(a+1),g(a-1)\notin\tau\Fthree.
\]
Therefore
\[
  A=g(a+1)-\tau\notin\tau\Fthree,
  \qquad
  B=g(a-1)-\tau\notin\tau\Fthree.
\]
Applying Lemma~\ref{lem:quadratic-identities} with $(x,y)=(a,1)$,
\[
  A+B=g(a+1)+g(a-1)-2\tau=(-g(a)-\tau)-2\tau=-g(a).
\]
Since $a\notin\Fthree$, Corollary~\ref{cor:special-preimages} gives $g(a)\notin\tau\Fthree$,
and hence $A+B\notin\tau\Fthree$.  Finally,
\[
  A-B=g(a+1)-g(a-1)=-L_1(a)
\]
by Lemma~\ref{lem:quadratic-identities} with $(x,y)=(a,1)$.  Since $g$ is PN,
$L_1$ is injective.  Hence Lemma~\ref{lem:L1-image} and
$a\notin\Fthree$ give $L_1(a)\notin\tau\Fthree$, and therefore
$A-B\notin\tau\Fthree$.

The preceding argument shows that
\[
  |D_aG(E_a)|=|E_a|=4,
  \qquad
  D_aG(E_a)\cap D_ag(E_a)=\varnothing.
\]
Hence the restriction of $D_aG$ to $E_a$ is injective.  Since $D_ag$ is
a permutation and $D_aG=D_ag$ on $\Fq\setminus E_a$, the restriction of
$D_aG$ to $\Fq\setminus E_a$ is also injective, and
\[
  D_aG(\Fq\setminus E_a)
  =D_ag(\Fq\setminus E_a)
  =\Fq\setminus D_ag(E_a).
\]
In particular,
\[
  D_aG(E_a)\subseteq D_aG(\Fq\setminus E_a).
\]
Consequently,
\[
  \delta_G(a,b)=
  \begin{cases}
    0,& b\in D_ag(E_a),\\
    2,& b\in D_aG(E_a),\\
    1,& \text{otherwise}.
  \end{cases}
\]
Thus the $a$-row of the DDT has type
\[
  0^4 1^{q-8}2^4.
\]
This proves the asserted row structure.  Summing the row types gives the
displayed differential spectrum and shows that $\delta_G=2$.
\end{proof}

For later use, we record the following consequences of Case~2.

\begin{corollary}
\label{cor:exceptional-values}
Let $a\in\Fq^*\setminus\Fthree$, and put
\[
  A:=D_ag(1)=g(a+1)-\tau,
  \qquad
  B:=D_ag(2)=g(a-1)-\tau.
\]
Then
\[
  A,\ B,\ A+B,\ A-B\notin\tau\Fthree.
\]
Moreover,
\[
  \{b\in\Fq:\delta_G(a,b)=0\}
  =D_ag(E_a)=\{\pm A,\pm B\}.
\]
The set $D_aG(E_a)$ has four elements and is disjoint from $D_ag(E_a)$.
\end{corollary}

\begin{proof}
All the assertions were established in Case~2 of the proof of
Theorem~\ref{thm:diff-spectrum}.
\end{proof}

The proof also gives the two-element fibers of $D_aG$ explicitly.

\begin{corollary}\label{cor:derivative-two-element-fibers}
Let $a\in\Fq^*\setminus\Fthree$, and define
\[
  u_a:=L_a^{-1}(-\tau).
\]
Then $u_a\notin\Fthree$, and the four two-element fibers of $D_aG$ are
\[
  \{1,1+u_a\},\qquad
  \{2,2+u_a\},\qquad
  \{1-a,1-a-u_a\},\qquad
  \{2-a,2-a-u_a\}.
\]
\end{corollary}

\begin{proof}
Since $L_a$ is bijective, $u_a$ is well defined.  Also $u_a\ne0$, because $-\tau\ne0$.
If $u_a\in\Fthree^*$, then, using the $\Fthree$-bilinearity and symmetry of
$(a,x)\mapsto L_a(x)$,
\[
  -\tau=L_a(u_a)=u_aL_a(1)=u_aL_1(a).
\]
Hence $L_1(a)\in\tau\Fthree=L_1(\Fthree)$ by
Lemma~\ref{lem:L1-image}.  Since $L_1$ is injective, this contradicts
$a\notin\Fthree$.  Thus $u_a\notin\Fthree$.

By Corollary~\ref{cor:exceptional-values}, the two exceptional
derivative-value sets in \eqref{eq:exceptional-derivative-value-sets} are
disjoint.

The first two displayed fibers can be treated simultaneously.  For each
$e\in\{1,2\}$,
\[
  D_ag(e+u_a)
  =D_ag(e)+L_a(u_a)
  =D_ag(e)-\tau
  =D_aG(e).
\]
The value $D_aG(e)$ belongs to $D_aG(E_a)$ and therefore, by this
disjointness, does not belong to $D_ag(E_a)$.  Since $D_ag$ is a
permutation, its unique preimage $e+u_a$ lies outside $E_a$.  Consequently,
$D_aG(e+u_a)=D_ag(e+u_a)=D_aG(e)$, so $\{e,e+u_a\}$ is a two-element
fiber of $D_aG$ for each $e\in\{1,2\}$.

The remaining two displayed fibers follow from the derivative symmetry of
the even map $G$.  By Lemma~\ref{lem:even-derivative}, the involution
$\iota_a(x)=-x-a$ sends fibers of $D_aG$ to fibers
of $D_aG$.  Since
\[
  \iota_a(2)=1-a,
  \qquad
  \iota_a(2+u_a)=1-a-u_a,
\]
and
\[
  \iota_a(1)=2-a,
  \qquad
  \iota_a(1+u_a)=2-a-u_a,
\]
we obtain the remaining two-element fibers
\[
  \{1-a,1-a-u_a\},
  \qquad
  \{2-a,2-a-u_a\}.
\]
The corresponding four $D_aG$-values are the four distinct elements of
$D_aG(E_a)$.  Theorem~\ref{thm:diff-spectrum} shows that $D_aG$ has exactly
four two-element fibers, so the displayed list is complete.
\end{proof}

\subsection{The exact boomerang spectrum}\label{sec:boomerang}

We now determine the boomerang rows directly from the fiber structure of
$D_aG$ obtained above.

\begin{theorem}[Exact boomerang spectrum]\label{thm:bct-spectrum}
Let $G$ be as in Theorem~\ref{thm:diff-spectrum}.  Then the following hold.
\begin{enumerate}[label=\textup{(\roman*)}]
\item If $a\in\Fthree^*$, then the $a$-row of the boomerang connectivity
table is a zero row:
\[
  0^{q-1}.
\]
\item If $a\in\Fq^*\setminus\Fthree$, define
\[
  u_a:=L_a^{-1}(-\tau)
\]
and
\[
  b_{a,1}:=G(1+u_a)-G(1),
  \qquad
  b_{a,2}:=G(2+u_a)-G(2).
\]
Then the four elements
\[
  \pm b_{a,1},\qquad \pm b_{a,2}
\]
are nonzero and pairwise distinct.  The four nonzero entries in the
$a$-row occur precisely at these four values of $b$, and every one of
them is equal to $2$.  Hence the row has type
\[
  0^{q-5}2^4.
\]
\end{enumerate}
Thus the BCT row structure of $G$ consists of two zero rows and $q-3$ rows of
type $0^{q-5}2^4$.
Consequently, $G$ has boomerang uniformity $\beta_G=2$, and its boomerang
spectrum is
\[
  \mathcal S_{\BCT}(G)
  =
  \{\!\{\,
    0^{(q-1)^2-4(q-3)},\,
    2^{4(q-3)}
  \,\}\!\}.
\]
\end{theorem}

\begin{proof}
For any map $F:\Fq\to\Fq$ and any $a,b\in\Fq^*$, the boomerang equations are equivalent to
\begin{equation}\label{eq:bct-fiber-form}
  D_aF(x)=D_aF(y),
  \qquad
  F(y)-F(x)=b.
\end{equation}
Indeed, the two equations
\[
  F(y+a)-F(x+a)=b,
  \qquad
  F(y)-F(x)=b
\]
are equivalent to the second equation and
\[
  F(y+a)-F(y)=F(x+a)-F(x).
\]

If $a\in\Fthree^*$, then $D_aG$ is a permutation by Theorem~\ref{thm:diff-spectrum}.  Thus
\eqref{eq:bct-fiber-form} forces $x=y$, which is impossible because $b=G(y)-G(x)\ne0$.
Therefore
\[
  \beta_G(a,b)=0
  \qquad(a\in\Fthree^*,\ b\in\Fq^*).
\]
This proves the zero-row statement.

Now let $a\in\Fq^*\setminus\Fthree$, and put $u=u_a=L_a^{-1}(-\tau)$.  By Corollary~\ref{cor:derivative-two-element-fibers}, the only ordered pairs $(x,y)$ with $x\ne y$ and
$D_aG(x)=D_aG(y)$ come from the four two-element fibers of $D_aG$:
\[
  \{1,1+u\},\qquad
  \{2,2+u\},\qquad
  \{1-a,1-a-u\},\qquad
  \{2-a,2-a-u\}.
\]
Define
\[
  b_1:=G(1+u)-G(1),
  \qquad
  b_2:=G(2+u)-G(2).
\]
Since $u\notin\Fthree$, both $1+u$ and $2+u$ lie outside $\Fthree$, and hence
\[
  b_1=g(1+u)+\tau,
  \qquad
  b_2=g(2+u)+\tau=g(u-1)+\tau.
\]
By Corollary~\ref{cor:special-preimages}, neither $g(1+u)$ nor $g(2+u)$ is equal to
$-\tau$.  Hence
\begin{equation}\label{eq:b-nonzero-general}
  b_1,b_2\ne0.
\end{equation}

The first two-element fiber contributes the two ordered differences $b_1$ and
$-b_1$; the second two-element fiber
contributes $b_2$ and $-b_2$.  By Lemma~\ref{lem:even-derivative}, if
$D_aG(x)=D_aG(y)$, then
\[
  G(\iota_a(y))-G(\iota_a(x))=G(y)-G(x).
\]
Thus the involution
\[
  T(x,y):=(\iota_a(x),\iota_a(y))=(-x-a,-y-a)
\]
preserves the boomerang value on ordered pairs with equal $D_aG$-value.  The
remaining two-element fibers are obtained from the first two by this
involution.  Consequently the four two-element fibers contribute only the
four possible values
\[
  \pm b_1,
  \qquad
  \pm b_2,
\]
and each value will occur exactly twice once we know that these four values are distinct.

It remains to prove
\[
  b_1\ne\pm b_2.
\]
First,
\[
  b_1-b_2=g(u+1)-g(u-1)=-L_1(u)
\]
by Lemma~\ref{lem:quadratic-identities} with $(x,y)=(u,1)$.  Since $L_1$ is bijective and $u\ne0$, this is
nonzero.  Therefore $b_1\ne b_2$.

For the nontrivial cancellation, Lemma~\ref{lem:quadratic-identities} with $(x,y)=(u,1)$ gives
\[
\begin{aligned}
  b_1+b_2
   &=g(u+1)+g(u-1)+2\tau  \\
   &=(-g(u)-\tau)+2\tau  \\
   &=\tau-g(u).
\end{aligned}
\]
If $b_1+b_2=0$, then $g(u)=\tau$.  By Corollary~\ref{cor:special-preimages},
$g^{-1}(\tau)=\{\pm1\}$, so $u\in\Fthree$, contradicting Corollary~\ref{cor:derivative-two-element-fibers}.
Thus $b_1+b_2\ne0$, and hence $b_1\ne-b_2$.

Therefore the four nonzero values
\[
  b_1,
  \quad -b_1,
  \quad b_2,
  \quad -b_2
\]
are distinct.  Each occurs exactly twice in the $a$-row of the boomerang table,
and no other nonzero $b$ occurs.  Hence this row has type
\[
  0^{q-5}2^4.
\]
This proves the asserted row structure.  Summing the row types gives the
displayed boomerang spectrum and shows that $\beta_G=2$.
\end{proof}

\subsection{Comparison with power functions and Ness--Helleseth-type binomials}
\label{subsec:spectral-comparison}

We now compare the differential and boomerang spectra of $G$ with those of two broad classes
for which, in both the DDT and BCT, all rows indexed by $a\ne0$ have the same value
distribution.  The following terminology
excludes the degenerate monomial parameter $u=0$.

\begin{definition*}[Ness--Helleseth-type binomials]
Let $r$ be a positive integer and let $u\in\Fq^*$.  Let
\[
  \chi(x):=x^{(q-1)/2}
  \qquad(x\in\Fq)
\]
be the quadratic character of $\Fq$, with $\chi(0)=0$.  The binomial function
\[
  F_{r,u}(x):=x^r\bigl(1+u\chi(x)\bigr)
\]
is called a \emph{Ness--Helleseth-type binomial} on $\Fq$.
\end{definition*}

The original Ness--Helleseth family is obtained over $\Fq=\F_{3^n}$, with $n\ge3$
odd, by taking $r=q-2$.  As polynomial functions on $\Fq$,
\[
  F_{q-2,u}(x)=x^{q-2}+u x^{(q-3)/2}.
\]
Ness and Helleseth \cite[Thm.~1]{NH07} proved that if
$\chi(u-1)=\chi(u+1)=\chi(u)$, then $F_{q-2,u}$ is APN.  Xia et al.\
\cite[Thm.~4]{XiaEtAl24} later proved the converse.

For primes $p\ge7$ with $p\equiv3\pmod4$ and odd $m$, Zeng et al.
extended the construction to $\F_{p^m}$.  With $\chi$ denoting the
quadratic character of $\F_{p^m}$ and $u\in\F_{p^m}^*$, the function
$F_{p^m-2,u}$ is APN if either $\chi(u+1)=\chi(u-1)=-\chi(5u+3)$ or
$\chi(u+1)=\chi(u-1)=-\chi(5u-3)$; see \cite{ZengHuYangJiang07}.

The following row-reduction formulas are given in Mesnager and Wu
\cite[Lemma~11]{MW25}.

\begin{lemma}
\label{lem:NH-row-reduction}
Assume that $q=3^n$ with $n>1$ odd.  Let
\[
  F_{r,u}(x)=x^r\bigl(1+u\chi(x)\bigr),
  \qquad u\in\Fq^*.
\]
Then, for every $a\in\Fq^*$ and $b\in\Fq$,
\[
  \delta_{F_{r,u}}(a,b)=
  \begin{cases}
    \displaystyle \delta_{F_{r,u}}\left(1,\frac{b}{a^r}\right),
      & \chi(a)=1,\\[3mm]
    \displaystyle \delta_{F_{r,u}}\left(1,\frac{b}{(-1)^{r+1}a^r}\right),
      & \chi(a)=-1,
  \end{cases}
\]
and, for every $a,b\in\Fq^*$,
\[
  \beta_{F_{r,u}}(a,b)=
  \begin{cases}
    \displaystyle \beta_{F_{r,u}}\left(1,\frac{b}{a^r}\right),
      & \chi(a)=1,\\[3mm]
    \displaystyle \beta_{F_{r,u}}\left(1,\frac{b}{(-1)^r a^r}\right),
      & \chi(a)=-1.
  \end{cases}
\]
Consequently, all DDT rows indexed by $a\in\Fq^*$ have the same value distribution, and all
BCT rows indexed by $a\in\Fq^*$ have the same value distribution.
\end{lemma}

\begin{theorem}[CCZ-inequivalence of $G$ to power functions and
Ness--Helleseth-type binomials]
\label{thm:spectral-separation}
Assume that $q=3^n$ with $n>1$ odd, and let $G$ be as in
Theorem~\ref{thm:diff-spectrum}.  If $H:\Fq\to\Fq$ is a power function or a
Ness--Helleseth-type binomial, then
\[
  \mathcal S_{\DDT}(G)\ne\mathcal S_{\DDT}(H)
  \qquad\text{and}\qquad
  \mathcal S_{\BCT}(G)\ne\mathcal S_{\BCT}(H).
\]
Consequently, by the CCZ-invariance of the differential spectrum,
$G\not\sim_{\CCZ}H$.
\end{theorem}

\begin{proof}
For a power function $H(x)=x^d$, the substitutions $x=aX$ in the derivative equation and
$(x,y)=(aX,aY)$ in the boomerang system give
\[
  \delta_H(a,b)=\delta_H\left(1,\frac{b}{a^d}\right),
  \qquad
  \beta_H(a,b)=\beta_H\left(1,\frac{b}{a^d}\right).
\]
For a Ness--Helleseth-type binomial, the corresponding conclusions follow from
Lemma~\ref{lem:NH-row-reduction}.  Hence, in either case, every DDT row indexed by
$a\in\Fq^*$ has one common value distribution and every BCT row indexed by $a\in\Fq^*$
has one common value distribution.

For $i\ge0$, put
\[
  \omega_i(H):=\#\{b\in\Fq:\delta_H(1,b)=i\},
  \qquad
  \nu_i(H):=\#\{b\in\Fq^*:\beta_H(1,b)=i\}.
\]
Then
\[
  N_0^{\DDT}(H)=(q-1)\omega_0(H),
  \qquad
  N_2^{\BCT}(H)=(q-1)\nu_2(H),
\]
so both numbers are divisible by $q-1$.  On the other hand,
Theorems~\ref{thm:diff-spectrum} and \ref{thm:bct-spectrum} give, respectively,
\[
  N_0^{\DDT}(G)=4(q-3),
  \qquad
  N_2^{\BCT}(G)=4(q-3).
\]
But
\[
  4(q-3)=4(q-1)-8.
\]
If $q-1$ divided $4(q-3)$, then $q-1$ would divide $8$.  Since $q=3^n$ with
$n>1$ odd, one has $q-1\ge26$, which is impossible.  Thus both spectra are
different.  The differential-spectrum inequality and
Lemma~\ref{lem:ccz-global-ddt-spectrum} then give $G\not\sim_{\CCZ}H$.
\end{proof}

\subsubsection*{Selected APN monomial families and DDT-row distributions}

For reference, Table~\ref{tab:apn-monomial-row-distributions} summarizes selected APN power
functions over $\Fq$ and the known status of their DDT-row distributions.  The uppercase
labels $F_i$ are local to the table and are unrelated to the lower-case PN-family labels
$f_i$ in Subsection~\ref{sec:known-examples}.  With $\omega_i(F)$ as defined above,
$\bigl(\omega_0(F),\omega_1(F),\omega_2(F)\bigr)$ denotes the common value distribution of
the DDT rows indexed by $a\ne0$ for an APN monomial $F$.  In the $F_3$ row, $\chi$ denotes
the quadratic character of $\Fq$, and
$\lambda_{3,n}$ is specified by
\[
  \begin{gathered}
    \lambda_{3,n}:=2\sum_{x\in\Fq}\chi\bigl(x(x^2+x-1)\bigr),\\[-0.2ex]
    \lambda_{3,1}=\lambda_{3,2}=4,
    \qquad
    \lambda_{3,n}=-2\lambda_{3,n-1}-3\lambda_{3,n-2}\quad(n\ge3).
  \end{gathered}
\]

\begin{table}[!htbp]
\centering
\begingroup
\scriptsize
\setlength{\tabcolsep}{2.6pt}
\renewcommand{\arraystretch}{1.13}
\begin{tabular}{@{}>{\raggedright\arraybackslash}p{0.065\textwidth}
                    >{\raggedright\arraybackslash}p{0.41\textwidth}
                    >{\raggedright\arraybackslash}p{0.345\textwidth}
                    >{\centering\arraybackslash}p{0.11\textwidth}@{}}
\toprule
Family & Monomial $F_i(x)=x^{d_i}$ over $\F_{3^n}$ & DDT-row distribution & \shortstack{Source/\\status} \\
\midrule
$F_1$
& $\displaystyle d_1=\frac{3^{n+1}-1}{4}$
& $\displaystyle \omega_0=\omega_2=\frac{q-3}{2},\quad \omega_1=3$
& \cite{CHNC13} \\
\addlinespace[0.8mm]
$F_2$
& $\displaystyle d_2=q-3=3^n-3$
& $\displaystyle \omega_0=\omega_2=\frac{q-3}{2},\quad \omega_1=3$
& \cite{XZLH20,YXLHL22} \\
\addlinespace[0.8mm]
\shortstack[l]{$F_3$\\$(n\ge5)$}
& $\displaystyle d_3=\frac{q-3}{2}=\frac{3^n-3}{2}$
& $\displaystyle
  \begin{aligned}
    \omega_0&=\omega_2=\frac{q+\lambda_{3,n}-7}{4},\\[-0.2ex]
    \omega_1&=\frac{q-\lambda_{3,n}+7}{2}
  \end{aligned}$
& \cite{YMT24} \\
\addlinespace[0.8mm]
\shortstack[l]{$F_4$\\$(n\ge5)$}
& $\displaystyle d_4=\begin{cases}
\dfrac{3^{(n+1)/2}-1}{2}, & n\equiv3\pmod4,\\[2mm]
\dfrac{3^{(n+1)/2}-1}{2}+\dfrac{q-1}{2}, & n\equiv1\pmod4
\end{cases}$
& \emph{Unknown in general}
& \cite{DMM03,Leducq12} \\
\addlinespace[0.8mm]
\shortstack[l]{$F_5$\\$(n\ge5)$}
& $\displaystyle d_5=\begin{cases}
\dfrac{3^{n+1}-1}{8}, & n\equiv3\pmod4,\\[2mm]
\dfrac{3^{n+1}-1}{8}+\dfrac{q-1}{2}, & n\equiv1\pmod4
\end{cases}$
& \emph{Unknown in general}
& \cite{DMM03,Leducq12} \\
\addlinespace[0.8mm]
$F_6$
& $\displaystyle d_6=\frac{3^{n+1}-1}{3^{(n+1)/2^\ell}+1},\quad
   \ell\ge2,\quad n\equiv-1\pmod{2^\ell}$
& $\displaystyle \omega_0=\omega_2=\frac{q-3}{2},\quad \omega_1=3$
& \cite{KKKB26,Leducq12,ZW10} \\
\addlinespace[0.8mm]
$F_7$
& {$\displaystyle \text{even }d_7\ \text{with}\ d_7(3^m+1)\equiv2\pmod{3^n-1},$

  \hspace{1.1em}$\displaystyle 1\le m\le\frac{n-1}{2},\quad \gcd(m,n)=1$}
& $\displaystyle \omega_0=\omega_2=\frac{q-3}{2},\quad \omega_1=3$
& \cite{KKKB26,ZW10} \\
\bottomrule
\end{tabular}
\endgroup
\caption{Selected APN monomial families over $\F_{3^n}$, where $n\ge3$ is odd, and the status of their DDT-row distributions.}
\label{tab:apn-monomial-row-distributions}
\end{table}
\FloatBarrier

\paragraph{Brief notes on Table~\ref{tab:apn-monomial-row-distributions}.}
\begin{enumerate}[label=\textup{(\roman*)},leftmargin=*,nosep]
\item The DDT-row distributions of $F_4$ and $F_5$ have been computationally verified to
equal that of $F_3$ for every odd $5\le n\le11$.

\item For each admissible $m$, the congruence in the $F_7$ row uniquely determines the even
residue $d_7$ modulo $3^n-1$; see \cite[Theorem~4.1]{ZW10} for the APN criterion and
\cite[Section~3]{KKKB26} for an explicit parametrization and the DDT-row distribution.

\item The family $F_6$ is contained in $F_7$ via $m=(n+1)/2^\ell$, and $F_6=F_1$ when
$n+1=2^\ell$, equivalently $m=1$.  Since $\gcd(d_i,3^n-1)=2$ for every listed exponent,
Dempwolff's criterion \cite[Theorem~1.1]{Dempwolff18}, with the correction
\cite{Dempwolff22}, shows that $F_i\sim_{\CCZ}F_j$ if and only if
$d_j\equiv3^a d_i\pmod{3^n-1}$ for some $0\le a<n$.  Equivalently, two listed monomials
are CCZ-inequivalent precisely when their exponents lie in distinct $3$-cyclotomic cosets
modulo $3^n-1$.

\item Theorem~\ref{thm:spectral-separation} already proves, without computation, that the
boomerang spectrum of $G$ differs from that of every listed monomial.  Computations further give
$\beta_{F_i}\ge4$ for \mbox{$1\le i\le6$} and every odd \mbox{$5\le n\le11$} for which $F_i$ is defined;
at $n=11$, $(\beta_{F_1},\ldots,\beta_{F_6})=(8,8,4,8,6,10)$, whereas $\beta_G=2$ by
Subsection~\ref{sec:boomerang}.

\item For the original Ness--Helleseth APN family
$F_{3^n-2,u}(x)=x^{3^n-2}\bigl(1+u\chi(x)\bigr)$, where
\mbox{$\chi(u-1)=\chi(u+1)=\chi(u)$}, exhaustive computations over $\F_{3^n}$ for
$n\in\{9,11,13\}$ give
\mbox{$\beta_{F_{3^n-2,u}}=4$} for every admissible $u\in\F_{3^n}^*$, again exceeding
$\beta_G=2$.
\end{enumerate}
\smallskip

\begin{remark}
\label{rem:apn-beta-two}
Let $p$ be an odd prime and let $n$ be a positive integer.
Pal and St\u{a}nic\u{a}~\cite[Theorem~2.1 and Remark~2.3]{PalStanica25}
show that the boomerang uniformity of an odd APN permutation equals its
$(-1)$-differential uniformity.  Combining this bridge with their calculation for
the inverse permutation~\cite[Theorem~2.5]{PalStanica25} gives
\[
  I(x):=x^{p^n-2},
  \qquad
  \delta_I=\beta_I=2,
\]
whenever
\[
  \chi(-3)=-1
  \qquad\text{and}\qquad
  \chi(5)\ne1.
\]
These conditions yield an infinite family.  For example, they hold over $\F_{5^n}$
for every odd $n$: in this case $-3=2$ is a nonsquare in $\F_{5^n}$, whereas
$\chi(5)=0$.

The inverse family does not give an APN example in characteristic $3$, because then
$-3=0$ and hence $\chi(-3)=0$.  In contrast,
Theorems~\ref{thm:diff-spectrum} and~\ref{thm:bct-spectrum} give
$\delta_G=\beta_G=2$ over $\F_{3^n}$ for every odd $n>1$ and every
DO PN input $g$.
Thus, to the best of our knowledge, the present construction is the first infinite family
in characteristic $3$ whose differential and boomerang uniformities are both equal to $2$.
More broadly, prior to the present work, the inverse family appears to have been the only
established infinite family of APN functions over finite fields of odd characteristic
with boomerang uniformity two.
\end{remark}

\subsection{A lower bound on the boomerang uniformity of APN functions}
\label{sec:apn-beta-lower-bound}

Among the small-field APN binomials recorded by Budaghyan and Pal
\cite[Table~5]{BudaghyanPal25}, exhaustive computation reveals isolated
examples with boomerang uniformity one.  For instance,
\[
  \delta_{\,x^7+2x^2}=2,
  \qquad
  \beta_{\,x^7+2x^2}=1
  \qquad\text{over }\F_{11}.
\]
However, Budaghyan and Pal's conjecture that this binomial class contains an
infinite APN subfamily was subsequently disproved in
\cite{BartoliStanica26,MW25}.  Thus these isolated examples do not provide a general
construction or an infinite family.  We first show that boomerang uniformity
one is the minimum possible for an APN function.

\begin{proposition}\label{prop:beta-zero-pn}
Let $q$ be an odd prime power and let $F:\Fq\to\Fq$.  If
$\delta_F\le2$ and $\beta_F=0$, then $\delta_F=1$, that is, $F$ is
perfect nonlinear.  Conversely, if $F$ is perfect nonlinear, then
$\beta_F=0$.
\end{proposition}

\begin{proof}
Assume first that $\delta_F\le2$ and $\beta_F=0$, and fix
$h\in\Fq^*$.  We claim that every nonzero value of $D_hF$ occurs at
most once.  Indeed, suppose that
\[
  D_hF(r)=D_hF(s)=b\ne0
\]
for distinct $r,s\in\Fq$.  Put
\[
  a=s-r\ne0,
  \qquad
  (x,y)=(r,r+h).
\]
Then
\[
  F(y)-F(x)=D_hF(r)=b
\]
and
\[
  F(y+a)-F(x+a)=D_hF(s)=b.
\]
Thus $(x,y)$ is counted by $\beta_F(a,b)$, contradicting
$\beta_F=0$.  This proves the claim.

Put $m=\delta_F(h,0)$.  Since $\delta_F\le2$, one has
$m\in\{0,1,2\}$.  If $m=0$, then $D_hF$ would map its $q$ inputs
injectively into the $q-1$ nonzero field elements, which is impossible.
Suppose that $m=2$.  By the claim, the remaining $q-2$ inputs give
$q-2$ distinct nonzero derivative values, and hence these values are
$\Fq^*\setminus\{\mu\}$ for some $\mu\in\Fq^*$.  Since $q>2$, one has
\[
  \sum_{t\in\Fq^*}t=0,
\]
and consequently
\[
  \sum_{x\in\Fq}D_hF(x)
  =\sum_{t\in\Fq^*\setminus\{\mu\}}t
  =-\mu\ne0.
\]
On the other hand, translation invariance gives
\[
  \sum_{x\in\Fq}D_hF(x)
  =\sum_{x\in\Fq}F(x+h)-\sum_{x\in\Fq}F(x)
  =0,
\]
a contradiction.  Hence $m=1$.  The claim now shows that $D_hF$
takes $0$ exactly once and every nonzero field element exactly once.
Thus $D_hF$ is a permutation of $\Fq$.  Since $h\ne0$ was arbitrary,
$F$ is perfect nonlinear and $\delta_F=1$.

Conversely, suppose that $F$ is perfect nonlinear and that
$\beta_F(a,b)\ge1$ for some $a,b\in\Fq^*$.  Then there exist
$x,y\in\Fq$ such that
\[
  F(y+a)-F(x+a)=b=F(y)-F(x).
\]
It follows that $D_aF(y)=D_aF(x)$.  Since $D_aF$ is a permutation,
one has $x=y$, contradicting $b\ne0$.  Hence $\beta_F=0$.
\end{proof}

\begin{corollary}\label{cor:apn-beta-lower-bound}
Every APN function $F:\Fq\to\Fq$ over a finite field of odd
characteristic satisfies $\beta_F\ge1$.  Equivalently, an APN function
cannot have boomerang uniformity zero.
\end{corollary}

\begin{proof}
If an APN function had $\beta_F=0$, Proposition~\ref{prop:beta-zero-pn}
would imply that it is perfect nonlinear and hence has differential
uniformity $1$, contradicting the APN hypothesis.
\end{proof}

In Section~\ref{sec:affine-perturbations}, we attain this optimal value over
\(\F_{3^n}\), with \(n>1\) odd, by constructing large classes of affine
perturbations that remain APN.  To the best of our knowledge, the present work
provides the first general construction yielding infinite families of APN
functions with boomerang uniformity one.

\section[APN functions with boomerang uniformity one via affine perturbations]
{APN functions with boomerang uniformity one\\via affine perturbations}
\label{sec:affine-perturbations}

This section develops the affine-perturbation construction and applies it to
known DO PN families.

\subsection{The affine-perturbation construction}
\label{sec:affine-construction}

Starting from the sign-switched APN map $G$, we determine exactly which affine perturbations
reduce its boomerang uniformity from $2$ to the optimal value $1$, and then
isolate a natural translated family whose admissible parameters can be
counted exactly.  Let $A:\Fq\to\Fq$ be an arbitrary $\Fthree$-affine map,
equivalently, a polynomial function of algebraic degree at most one.  Write
uniquely
\[
  A(x)=M(x)+d,
  \qquad
  d=A(0),
  \qquad
  M(x)=A(x)-A(0)=\sum_{i=0}^{n-1}m_i x^{3^i},
\]
where $m_i\in\Fq$ and $M$ is the $\Fthree$-linear part of $A$, and put
\[
  G_A(x):=G(x)+A(x).
\]
The constant $d$ cancels from both derivatives and output differences.  Since
$G_A\sim_{\EA}G$, every $G_A$ is APN.  The next theorem gives an exact
affine criterion and the complete boomerang spectrum whenever the criterion
holds.

\begin{theorem}[Exact affine criterion and boomerang spectrum]
\label{thm:affine-criterion-spectrum}
For every $\Fthree$-affine map $A$ with linear part $M$,
\begin{equation}\label{eq:affine-avoidance}
  \beta_{G_A}=1
  \quad\Longleftrightarrow\quad
  M(u)\notin
  \{0,\,\pm L_1(u),\,\pm(\tau-g(u))\}
  \quad\text{for every }u\in\Fq\setminus\Fthree.
\end{equation}
If the avoidance condition on the right of
\eqref{eq:affine-avoidance} fails, then $\beta_{G_A}=2$.

Assume that the avoidance condition in \eqref{eq:affine-avoidance} holds and,
for $u\in\Fq\setminus\Fthree$, define
\[
  b_1(u):=G(1+u)-G(1)=g(u+1)+\tau,
  \qquad
  b_2(u):=G(2+u)-G(2)=g(u-1)+\tau.
\]
Define
\[
  \mathcal Z_M:=
  \{u\in\Fq\setminus\Fthree:
       M(u)\in\{\pm b_1(u),\,\pm b_2(u)\}\},
  \qquad
  z_M:=|\mathcal Z_M|.
\]
Then the exact BCT row distribution of $G_A$ is
\[
\begin{array}{c|c}
  \text{number of rows}&\text{BCT row type}\\ \hline
  2&0^{q-1}\\
  z_M&0^{q-7}1^6\\
  q-3-z_M&0^{q-9}1^8.
\end{array}
\]
Consequently, when $\beta_{G_A}=1$, the boomerang spectrum of $G_A$ is
\[
  \mathcal S_{\BCT}(G_A)
  =\{\!\{\,
      0^{(q-5)^2+2z_M},
      1^{8(q-3)-2z_M}
    \,\}\!\}.
\]
\end{theorem}

\begin{proof}
For every $a\ne0$,
\[
  D_aG_A(x)=D_aG(x)+M(a).
\]
Thus $D_aG_A$ and $D_aG$ have the same fibers.  If $a\in\Fthree^*$, then
$D_aG$, and hence $D_aG_A$, is a permutation by
Theorem~\ref{thm:diff-spectrum}\textup{(i)}; therefore the $a$-row of the BCT
of $G_A$ is zero
by \eqref{eq:bct-fiber-form}.

Let $a\in\Fq^*\setminus\Fthree$, and put
\[
  u:=u_a=L_a^{-1}(-\tau),
  \qquad
  t:=M(u).
\]
Corollary~\ref{cor:derivative-two-element-fibers} gives $u\notin\Fthree$, and the
identity above shows that the four nonsingleton fibers of $D_aG_A$ are
\[
  \{1,1+u\},\qquad \{2,2+u\},\qquad
  \{1-a,1-a-u\},\qquad \{2-a,2-a-u\}.
\]
Abbreviate $b_i:=b_i(u)$ for $i=1,2$.  Equation~\eqref{eq:b-nonzero-general}
gives $b_1,b_2\ne0$, and
Lemma~\ref{lem:quadratic-identities} gives
\[
  b_1-b_2=-L_1(u),
  \qquad
  b_1+b_2=\tau-g(u).
\]
Since $g$ is PN, $L_1$ is injective, so $L_1(u)\ne0$.  Also
$\tau-g(u)\ne0$ by Corollary~\ref{cor:special-preimages}.  Hence the four
elements $\pm b_1,\pm b_2$ are distinct.

Orient each displayed fiber from its first listed point to its second.  By the
last assertion of Lemma~\ref{lem:even-derivative}, the corresponding
$G$-differences are $b_1,b_2,b_2,b_1$, respectively.  Their increments are
$u,u,-u,-u$, so the affine perturbation contributes $t,t,-t,-t$,
respectively.  Thus define the four representative $G_A$-differences by
\[
\begin{aligned}
  d_1&:=G_A(1+u)-G_A(1)=b_1+t,\\
  d_2&:=G_A(2+u)-G_A(2)=b_2+t,\\
  d_3&:=G_A(1-a-u)-G_A(1-a)=b_2-t,\\
  d_4&:=G_A(2-a-u)-G_A(2-a)=b_1-t.
\end{aligned}
\]
Their negatives arise from the opposite orientations.  By
\eqref{eq:bct-fiber-form}, the nonzero BCT columns in the $a$-row are therefore
obtained from the signed multiset
\[
  \{\pm d_1,\pm d_2,\pm d_3,\pm d_4\},
\]
with every occurrence of $0$ omitted.

Comparing the four representatives up to sign gives exactly the collisions
listed below.  The four nonzero values of $t$ for which a collision occurs, namely
$\pm L_1(u)$ and $\pm(\tau-g(u))$, are pairwise distinct.  Indeed, using
$b_1-b_2=-L_1(u)$ and $b_1+b_2=\tau-g(u)$, an equality
$L_1(u)=\pm(\tau-g(u))$ would force $b_1=0$ or $b_2=0$.
Direct substitution in $d_1,d_2,d_3,d_4$ gives the exact row types:
\[
\begin{array}{c|c|c}
  \text{condition on }t&\text{collision}&\text{BCT row type}\\ \hline
  t=0&d_1=d_4,\quad d_2=d_3&0^{q-5}2^4\\
  t=-L_1(u)&d_1=d_3&0^{q-7}1^4 2^2\\
  t=L_1(u)&d_4=d_2&0^{q-7}1^4 2^2\\
  t=\tau-g(u)&d_1=-d_2&0^{q-7}1^4 2^2\\
  t=-(\tau-g(u))&d_4=-d_3&0^{q-7}1^4 2^2.
\end{array}
\]
For each of the four nonzero values of $t$ in the table, the six nonzero BCT
columns can be identified explicitly.  If $t=\pm L_1(u)$, then
\[
  \{b\in\Fq^*: \beta_{G_A}(a,b)=2\}
  =\{\pm(\tau-g(u))\},
\]
whereas if $t=\pm(\tau-g(u))$, then
\[
  \{b\in\Fq^*: \beta_{G_A}(a,b)=2\}
  =\{\pm L_1(u)\}.
\]
In either case,
\[
  \{b\in\Fq^*: \beta_{G_A}(a,b)=1\}
  =\{\pm b_1,\pm b_2\}.
\]
For example, when $t=L_1(u)=b_2-b_1$,
\[
  (d_1,d_2,d_3,d_4)
  =\bigl(b_2,-(\tau-g(u)),b_1,-(\tau-g(u))\bigr).
\]
All six column indices are nonzero and pairwise distinct.  No other collision
occurs, so no BCT entry can exceed $2$.

\Needspace{7\baselineskip}
It remains to translate the parameter $u=u_a$ back to the BCT row index $a$.
The identity $L_{u_a}(a)=L_a(u_a)=-\tau$, together with the bijectivity of
$L_{u_a}$, gives $u_{u_a}=a$.  Hence $a\mapsto u_a$ is an involution of
$\Fq\setminus\Fthree$.  Consequently, every $u$ violating the avoidance
condition in \eqref{eq:affine-avoidance} corresponds to the unique row index
$a=u_u$, and conversely.  The row types above therefore prove
\eqref{eq:affine-avoidance}, including the assertion that $\beta_{G_A}=2$ when
the avoidance condition fails.

Suppose now that the avoidance condition in \eqref{eq:affine-avoidance}
holds.  From the definitions of the four representatives,
\[
\begin{alignedat}{2}
  d_1=0&\quad\Longleftrightarrow\quad t=-b_1,
  &\qquad d_2=0&\quad\Longleftrightarrow\quad t=-b_2,\\
  d_3=0&\quad\Longleftrightarrow\quad t=b_2,
  &d_4=0&\quad\Longleftrightarrow\quad t=b_1.
\end{alignedat}
\]
Because the four values $\pm b_1,\pm b_2$ are distinct, exactly one representative vanishes
when $t\in\{\pm b_1,\pm b_2\}$, and none vanishes otherwise.  Thus the row
has type $0^{q-7}1^6$ in the former case and $0^{q-9}1^8$ in the latter.
The involution $a\mapsto u_a$ shows
that exactly $z_M$ directions $a\in\Fq^*\setminus\Fthree$ have row type
$0^{q-7}1^6$, and the remaining $q-3-z_M$ such directions have row type
$0^{q-9}1^8$.  Together with the two zero rows indexed by
$a\in\Fthree^*$, this proves the asserted row distribution.  Finally,
\[
\begin{aligned}
  N_1^{\BCT}(G_A)
    &=6z_M+8(q-3-z_M)=8(q-3)-2z_M,\\
  N_0^{\BCT}(G_A)
    &=(q-1)^2-N_1^{\BCT}(G_A)=(q-5)^2+2z_M,
\end{aligned}
\]
which completes the proof.
\end{proof}

\begin{corollary}
\label{cor:plus-minus-x-specializations}
For each of the two choices of $g$ below, one has $\tau=g(1)=1$.
Let $\sigma\in\{\pm1\}$, put $A_\sigma(x):=\sigma x$, and write
\[
  G_\sigma:=G_{A_\sigma};
  \qquad
  G_\sigma(x)=g(x)+\epsilon(x)+\sigma x.
\]
Then the following statements hold over $\Fq$.
\begin{enumerate}[label=\textup{(\alph*)}]
\item If $g(x)=x^4$, the Gold function~\cite{CM97,DO68}, then
$G_\sigma$ is APN with boomerang uniformity one, and its boomerang spectrum is
\[
  \mathcal S_{\BCT}(G_\sigma)
  =\{\!\{0^{(q-5)^2},1^{8(q-3)}\}\!\}.
\]
\item If $g(x)=x^{10}+x^6-x^2$ is the Ding--Yuan function with
$\lambda=-1$~\cite{DY06}, then $G_\sigma$ is APN
with boomerang uniformity one.  Let $z_n=0,6,10,16$ according as
$\gcd(n,15)=1,3,5,15$, respectively.  Then its boomerang spectrum is
\[
  \mathcal S_{\BCT}(G_\sigma)
  =\{\!\{
      0^{(q-5)^2+2z_n},
      1^{8(q-3)-2z_n}
    \}\!\}.
\]
\end{enumerate}
\end{corollary}

\begin{proof}
For $A_\sigma(x)=\sigma x$, one has $d=0$ and $M(u)=\sigma u$.  Since, for
each $u$, both the forbidden set in \eqref{eq:affine-avoidance} and the set
$\{\pm b_1(u),\pm b_2(u)\}$ defining $\mathcal Z_M$ are closed under
negation, replacing $M$ by $-M$ changes neither the avoidance condition in
\eqref{eq:affine-avoidance} nor $z_M$.  It therefore suffices to take
$M=\operatorname{id}$.

For either choice of $g$, put
\[
  b_s(X):=g(X+s)+1
  \qquad (s\in\{\pm1\}).
\]
Since both choices of $g$ are even, one has
$b_{-1}(X)=b_{+1}(-X)$.

First let $g(x)=x^4$.  Then $L_1(u)=u^3+u$.  The equalities
$u=\pm L_1(u)$ have only solutions in $\Fthree$, whereas
$u=\pm(1-u^4)$ gives $u^4+u-1=0$ or $u^4-u-1=0$.
Both quartics are irreducible over $\Fthree$, and hence have no roots in the
odd-degree extension $\Fq$.  Thus the avoidance condition in
\eqref{eq:affine-avoidance} holds.  To compute $z_M$, one has
\[
  b_{+1}(X)-X=X^4+X^3-1,
  \qquad
  b_{+1}(X)+X=(X-1)^3(X+1).
\]
The quartic $X^4+X^3-1$ is irreducible over $\Fthree$, while
$(X-1)^3(X+1)$ has only the roots $\pm1$.  The same conclusions hold after
replacing $X$ by $-X$, so $z_M=0$.  Part~\textup{(a)} now follows from
Theorem~\ref{thm:affine-criterion-spectrum}.

Now let $g(x)=x^{10}+x^6-x^2$.  Then
$L_1(X)=X^9-X^3-X$, and
\[
\begin{aligned}
  L_1(X)+X&=X^9-X^3=(X^3-X)^3,\\
  L_1(X)-X&=X(X^2+1)(X^6-X^4+X^2+1).
\end{aligned}
\]
The first polynomial has only the roots in $\Fthree$, while the quadratic and
sextic factors in the second line are irreducible over $\Fthree$.  The two
remaining equalities arising from the forbidden set in
\eqref{eq:affine-avoidance} are $u=1-g(u)$ and $u=-(1-g(u))$, which are
equivalent to $g(u)=1-u$ and $g(u)=1+u$, respectively.  For
$s\in\{\pm1\}$,
\[
  g(X)-1-sX
  =(X^4-sX^3+sX+1)(X^6+sX^5+X^4-X^2-1),
\]
where both factors are irreducible over $\Fthree$.  Thus every root of either
of the two equations $g(X)=1\pm X$ has even degree over $\Fthree$, so neither
equation has a root in $\Fq$.  Thus the avoidance condition in
\eqref{eq:affine-avoidance} holds.

It remains to compute $z_M$.  Direct factorization gives
\[
\begin{aligned}
  b_{+1}(X)-X
    &=(X^4-X-1)(X^6+X^5+X^3+X+1),\\
  b_{+1}(X)+X
    &=(X^2-1)(X^3-X^2+1)(X^5-X^4+1).
\end{aligned}
\]
The quartic, sextic, cubic, and quintic factors displayed here are
irreducible over $\Fthree$.  The irreducible quartic and sextic factors have
no roots in $\Fq$, while the roots $\pm1$ of $X^2-1$ are excluded from
$\mathcal Z_M$.  Moreover,
\[
  \gcd\bigl(b_{+1}(X)+X,b_{+1}(-X)-X\bigr)=X^2-1,
\]
so the remaining root sets obtained from $b_{+1}$ and $b_{-1}$ are disjoint.  Since
$q=3^n$, the two cubic factors contribute six roots precisely when
$3\mid n$, and the two quintic factors contribute ten roots precisely when
$5\mid n$.
Consequently, $z_M=z_n$.
Substitution into Theorem~\ref{thm:affine-criterion-spectrum} proves
part~\textup{(b)}.
\end{proof}

Corollary~\ref{cor:plus-minus-x-specializations} gives two examples for which the fixed
perturbations $A(x)=\pm x$ yield boomerang uniformity one.  Such perturbations
are not universal: for a general DO PN function $g$, the admissibility of $A$
depends on $g$ through the criterion in Theorem~\ref{thm:affine-criterion-spectrum}.
A natural $g$-dependent family is provided by the derivatives
\[
D_cg(x)=L_c(x)+g(c).
\]
The next proposition and
Corollary~\ref{cor:natural-derivative-perturbations} show that, for every DO PN
function $g$ under our standing assumptions, exactly $(q-3)/2$ parameters
$c\in\Fq$ give $\beta_{G+D_cg}=1$.

\begin{proposition}
\label{prop:admissible-lc}
Put
\[
  \mathcal C_g:=
  \{c\in\Fq\setminus\Fthree:
       g(c)+\tau\notin g(\Fq)\}.
\]
Then $|\mathcal C_g|=(q-3)/2$.  Moreover, for every $c,d\in\Fq$,
\[
  \beta_{G+L_c+d}=1
  \quad\Longleftrightarrow\quad
  c\in\mathcal C_g.
\]
In particular, taking $d=g(c)$, one obtains
\[
  G(x)+D_cg(x)=g(x+c)+\tau\epsilon(x).
\]
Consequently, the function $g(x+c)+\tau\epsilon(x)$ has boomerang
uniformity one if and only if $c\in\mathcal C_g$.

For each $c\in\Fq\setminus\Fthree$, define
\begin{equation}\label{eq:kappa-c}
\begin{aligned}
  \kappa_c
  &:=\#\{\sigma\in\{\pm1\}:
       g(c+\sigma)+\tau\in g(\Fq)\}\\
  &=\#\{\sigma\in\{\pm1\}:c+\sigma\notin\mathcal C_g\}
  \in\{0,1,2\}.
\end{aligned}
\end{equation}
For $c\in\mathcal C_g$, one has
\[
  z_{L_c}=4\kappa_c\in\{0,4,8\}.
\]
\end{proposition}

\begin{proof}
We first count $\mathcal C_g$.  Put $D:=g(\Fq^*)$.  We claim that exactly
$(q-3)/4$ elements $z\in D$ satisfy $z+\tau\in D$.  Since $g$ is PN,
there are exactly $q-1$ pairs $(w,c)\in\Fq^2$ satisfying
$g(w)-g(c)=\tau$: for each $a=w-c\in\Fq^*$, the equation
$D_ag(c)=\tau$ has a unique solution.  Exactly two of these pairs have
$c=0$, namely $(w,c)=(1,0),(-1,0)$, and none has $w=0$, since
$g^{-1}(-\tau)=\varnothing$ by Corollary~\ref{cor:special-preimages}.  Hence
there are $q-3$ such pairs in $(\Fq^*)^2$.  Since every element of $D$ has
exactly two preimages by Theorem~\ref{thm:CM-two-to-one}, each ordered pair
$(z+\tau,z)\in D^2$ has four ordered preimages under $g\times g$, which
proves the claim.
Since $-\tau\notin D$, one has $z+\tau\in g(\Fq)$ if and only if
$z+\tau\in D$.  Therefore the number
of nonzero $c$ satisfying $g(c)+\tau\notin g(\Fq)$ is
\[
  2\left(\frac{q-1}{2}-\frac{q-3}{4}\right)=\frac{q+1}{2}.
\]
The two elements $c=\pm1$ are among them, whereas $c=0$ is not.  Removing
$\Fthree$ gives $|\mathcal C_g|=(q-3)/2$.

We next prove the asserted equivalence.  For $M=L_c$, the choices $c=0,1,-1$
give $M(u)=0,L_1(u),-L_1(u)$, respectively, so each violates the avoidance
condition in \eqref{eq:affine-avoidance}.  Now suppose
$c\notin\Fthree$.  Since $c,c-1,c+1\ne0$, the PN property of $g$ implies
that $L_c,L_{c-1}$, and $L_{c+1}$ are bijective.  Therefore, for every
$u\ne0$, one has
\[
  L_c(u)\notin\{0,\pm L_1(u)\}.
\]
For every $u\in\Fq$, the two remaining equalities arising from the forbidden
set in \eqref{eq:affine-avoidance} are equivalent to
\[
\begin{aligned}
  L_c(u)=\tau-g(u)
    &\quad\Longleftrightarrow\quad
      g(c+u)=g(c)+\tau,\\
  L_c(u)=-(\tau-g(u))
    &\quad\Longleftrightarrow\quad
      g(c-u)=g(c)+\tau.
\end{aligned}
\]
Consequently, if $c\in\mathcal C_g$, then $L_c(u)$ does not lie in the
forbidden set in \eqref{eq:affine-avoidance} for any
$u\in\Fq\setminus\Fthree$, so the avoidance condition in
\eqref{eq:affine-avoidance} holds.  Conversely, still assuming
$c\notin\Fthree$,
suppose that $c\notin\mathcal C_g$.  Choose $w\in\Fq$ such that
$g(w)=g(c)+\tau$, and put $u:=w-c$.  Since $\tau\ne0$, one has $u\ne0$, and
the first displayed equivalence gives $L_c(u)=\tau-g(u)$.  If $u=\pm1$, then
$g(u)=\tau$, so $L_c(u)=0$, contradicting the bijectivity
of $L_c$.  Therefore $u\in\Fq\setminus\Fthree$, and the equality above
violates the avoidance condition on the right-hand side of
\eqref{eq:affine-avoidance} for $M=L_c$.  Hence
$\beta_{G+L_c+d}=2$.  This proves the asserted equivalence.  The translated
identity in the statement follows directly from the definition of $L_c$.

It remains to determine $z_{L_c}$.  Let $c\in\mathcal C_g$ and
$u\in\Fq\setminus\Fthree$.  Since $M=L_c$ satisfies the avoidance condition
of Theorem~\ref{thm:affine-criterion-spectrum}, the last part of its proof
shows that at most one of $d_1,d_2,d_3,d_4$ can vanish for a fixed $u$, and
that
\[
  u\in\mathcal Z_{L_c}
  \quad\Longleftrightarrow\quad
  \text{one of }d_1,d_2,d_3,d_4\text{ is zero}.
\]
Substituting $t=L_c(u)$ into the displayed formulas for $d_1,d_2,d_3,d_4$
in that proof, the bilinearity of $(a,x)\mapsto L_a(x)$ gives
\[
\begin{aligned}
  d_1&=g(c+1+u)-g(c+1)-\tau,\\
  d_2&=g(c-1+u)-g(c-1)-\tau,\\
  d_3&=g(c+1-u)-g(c+1)-\tau,\\
  d_4&=g(c-1-u)-g(c-1)-\tau.
\end{aligned}
\]

Consider first $e=c+1$.  If $g(e)+\tau\notin g(\Fq)$, neither $d_1$ nor
$d_3$ vanishes.  Suppose that $g(e)+\tau\in g(\Fq)$.  This image value is
nonzero, because $g^{-1}(-\tau)=\varnothing$ by
Corollary~\ref{cor:special-preimages}.  Since $g$ is a DO PN function,
Theorem~\ref{thm:CM-two-to-one} shows that
$d_1=0$ has exactly two solutions $u$ and $d_3=0$ has exactly two solutions
$u$.  None of these solutions lies in $\Fthree$: $u=0$ is impossible because
$\tau\ne0$, while $u=\pm1$ would give
$g(e+u)=g(e)+\tau=g(e)+g(u)$ and hence $L_e(u)=0$,
contradicting the bijectivity of $L_e$.  By the at-most-one property just
noted, the four solutions obtained from $d_1=0$ and $d_3=0$ are distinct
elements of $\Fq\setminus\Fthree$.

The same argument with $e=c-1$ shows that $d_2=0$ and $d_4=0$ together
give exactly four elements of $\Fq\setminus\Fthree$ if
$g(c-1)+\tau\in g(\Fq)$, and none otherwise.  These four elements are
disjoint from those obtained from $d_1=0$ and $d_3=0$, again because at most
one of $d_1,d_2,d_3,d_4$ can vanish for a fixed $u$.  Therefore each
$e\in\{c-1,c+1\}$ satisfying $g(e)+\tau\in g(\Fq)$ contributes exactly four
elements of $\mathcal Z_{L_c}$.  By \eqref{eq:kappa-c},
$z_{L_c}=4\kappa_c$.
\end{proof}

The preceding proposition now yields the following natural translated family.

\begin{corollary}[Main construction: APN functions with boomerang uniformity one]
\label{cor:natural-derivative-perturbations}
For each $c\in\Fq$, define
\[
  \widetilde G_c(x)
  :=G(x)+D_cg(x)
  =g(x+c)+\tau\epsilon(x).
\]
Then the differential and boomerang uniformities of $\widetilde G_c$ are
\[
  \delta_{\widetilde G_c}=2,
  \qquad
  \beta_{\widetilde G_c}=
  \begin{cases}
    1, & c\in\mathcal C_g,\\
    2, & c\in\Fq\setminus\mathcal C_g.
  \end{cases}
\]
Consequently, among the $q$ parameter values $c\in\Fq$, exactly
$(q-3)/2$ give boomerang uniformity $1$, and exactly $(q+3)/2$ give
boomerang uniformity $2$.
\end{corollary}

\begin{proof}
The displayed identity follows from
$D_cg(x)=g(x+c)-g(x)$ and \eqref{eq:sign-switch}.
Since $\widetilde G_c\sim_{\EA}G$, one has
$\delta_{\widetilde G_c}=\delta_G=2$.  Taking $d=g(c)$ in
Proposition~\ref{prop:admissible-lc} and applying
Theorem~\ref{thm:affine-criterion-spectrum} gives the stated values of
$\beta_{\widetilde G_c}$.  Finally,
$\lvert\Fq\setminus\mathcal C_g\rvert=q-(q-3)/2=(q+3)/2$.
\end{proof}

We next complete the main result by determining the boomerang spectrum for every
parameter, including the nonadmissible parameters and the three special
parameters $c\in\Fthree$.

\begin{corollary}[Complete boomerang spectra of the main family]
\label{cor:lc-bct-spectrum}
For each $c\in\Fq$, let $\widetilde G_c$ be as in
Corollary~\ref{cor:natural-derivative-perturbations}.  For
$c\in\Fq\setminus\Fthree$, let $\kappa_c$ be defined by \eqref{eq:kappa-c}.
\begin{enumerate}[label=\textup{(\alph*)}]
\item If $c\in\mathcal C_g$, the exact BCT row distribution of
$\widetilde G_c$ is
\[
\begin{array}{c|c}
  \text{number of rows}&\text{BCT row type}\\ \hline
  2&0^{q-1}\\
  4\kappa_c&0^{q-7}1^6\\
  q-3-4\kappa_c&0^{q-9}1^8.
\end{array}
\]
Consequently, the boomerang spectrum of $\widetilde G_c$ is
\[
  \mathcal S_{\BCT}(\widetilde G_c)
  =\{\!\{\,
      0^{(q-5)^2+8\kappa_c},
      1^{8(q-3)-8\kappa_c}
    \,\}\!\}.
\]
\item If $c\in\Fq\setminus\bigl(\Fthree\cup\mathcal C_g\bigr)$, the exact
BCT row distribution of $\widetilde G_c$ is
\[
\begin{array}{c|c}
  \text{number of rows}&\text{BCT row type}\\ \hline
  2&0^{q-1}\\
  4&0^{q-7}1^4 2^2\\
  4\kappa_c&0^{q-7}1^6\\
  q-7-4\kappa_c&0^{q-9}1^8.
\end{array}
\]
Consequently, the boomerang spectrum of $\widetilde G_c$ is
\[
  \mathcal S_{\BCT}(\widetilde G_c)
  =\{\!\{\,
      0^{(q-5)^2+8(\kappa_c+1)},
      1^{8(q-5-\kappa_c)},
      2^8
    \,\}\!\}.
\]

\item If $c=0$, the exact BCT row distribution of $\widetilde G_c$ is
\[
\begin{array}{c|c}
  \text{number of rows}&\text{BCT row type}\\ \hline
  2&0^{q-1}\\
  q-3&0^{q-5}2^4.
\end{array}
\]
Consequently, the boomerang spectrum of $\widetilde G_0$ is
\[
  \mathcal S_{\BCT}(\widetilde G_0)
  =\{\!\{\,
      0^{q^2-6q+13},
      2^{4(q-3)}
    \,\}\!\}.
\]

\item If $c\in\Fthree^*$, the exact BCT row distribution of
$\widetilde G_c$ is
\[
\begin{array}{c|c}
  \text{number of rows}&\text{BCT row type}\\ \hline
  2&0^{q-1}\\
  q-3&0^{q-7}1^4 2^2.
\end{array}
\]
Consequently, the boomerang spectrum of $\widetilde G_c$ is
\[
  \mathcal S_{\BCT}(\widetilde G_c)
  =\{\!\{\,
      0^{q^2-8q+19},
      1^{4(q-3)},
      2^{2(q-3)}
    \,\}\!\}.
\]
\end{enumerate}
\end{corollary}

\begin{proof}
Apply Theorem~\ref{thm:affine-criterion-spectrum} to
$\widetilde G_c=G_A$, where $A=D_cg=L_c+g(c)$ has linear part $L_c$.
Part \textup{(a)} follows by substituting
$z_{L_c}=4\kappa_c$ from Proposition~\ref{prop:admissible-lc} into
Theorem~\ref{thm:affine-criterion-spectrum}.

For part \textup{(b)}, suppose that
$c\in\Fq\setminus(\Fthree\cup\mathcal C_g)$.  The proof of
Proposition~\ref{prop:admissible-lc} shows that the avoidance condition in
\eqref{eq:affine-avoidance} can fail precisely when
\[
  g(c+u)=g(c)+\tau
  \qquad\text{or}\qquad
  g(c-u)=g(c)+\tau.
\]
Since $c\notin\mathcal C_g$, the common right-hand side is a nonzero element
of $g(\Fq)$ by Corollary~\ref{cor:special-preimages}.  By
Theorem~\ref{thm:CM-two-to-one}, each equation has exactly
two solutions.  As in the proof of Proposition~\ref{prop:admissible-lc}, none
lies in $\Fthree$; moreover, by the same equivalences the two pairs are
disjoint, since a common solution would give $g(u)=\tau$, and hence
$u=\pm1$.  Thus precisely four
parameter values $u\in\Fq\setminus\Fthree$ violate
\eqref{eq:affine-avoidance}.  For these four values,
$t=L_c(u)=\pm(\tau-g(u))$, so the corresponding cases in the collision table
in the proof of Theorem~\ref{thm:affine-criterion-spectrum} give type
$0^{q-7}1^4 2^2$ and no vanishing $d_i$.

For every $u\in\Fq\setminus\Fthree$, the four values
$\pm b_1(u),\pm b_2(u)$ are distinct, independently of the avoidance condition
in \eqref{eq:affine-avoidance},
as shown in the proof of Theorem~\ref{thm:affine-criterion-spectrum}.
Hence at most one of $d_1,d_2,d_3,d_4$ can vanish for a fixed $u$.
This allows us to repeat the counting argument from
Proposition~\ref{prop:admissible-lc} even when $c\notin\mathcal C_g$.
The same calculation in the final two paragraphs of the proof of
Proposition~\ref{prop:admissible-lc}, which uses only $c\notin\Fthree$,
shows that $d_1=0$ and $d_3=0$ have
altogether four distinct solutions $u$ precisely when
$g(c+1)+\tau\in g(\Fq)$, while $d_2=0$ and $d_4=0$ have altogether four
distinct solutions $u$ precisely when $g(c-1)+\tau\in g(\Fq)$.  Hence, by
\eqref{eq:kappa-c}, one of $d_1,d_2,d_3,d_4$ vanishes for exactly
$4\kappa_c$ values of $u\in\Fq\setminus\Fthree$.  Each such value gives
type $0^{q-7}1^6$.  The remaining $q-7-4\kappa_c$ parameter values give
type $0^{q-9}1^8$.  The involution
$a\mapsto u_a$ from the proof of
Theorem~\ref{thm:affine-criterion-spectrum} shows that these parameter
counts are exactly the corresponding BCT row counts.  Summing the row types gives the asserted
boomerang spectrum.

If $c=0$, then $\widetilde G_0=G$, so part \textup{(c)} follows from
Theorem~\ref{thm:bct-spectrum}.

Finally, let $c\in\Fthree^*$.  For $c=1$ and $c=-1$, respectively, the
third and second rows of the collision and BCT row-type table in the proof of
Theorem~\ref{thm:affine-criterion-spectrum} apply, since
$t=L_c(u)=cL_1(u)$.  Hence every row indexed by a direction
$a\in\Fq^*\setminus\Fthree$ has type
$0^{q-7}1^4 2^2$, which proves part \textup{(d)}.
\end{proof}

\begin{remark}[A Gold example realizing all eight spectrum types]
\label{rem:gold-all-kappa}
Let $g(x)=x^{10}=x^{3^2+1}$ on $\Fq$, where $q=3^n$ and $n\ge9$ is
odd.  This is a Gold DO PN function~\cite{CM97,DO68}.
Put
\[
  P_0(X):=X^{10}+1,
  \qquad
  P_{\pm1}(X):=(X\pm1)^{10}+1.
\]
The derivatives $P_0'=X^9$ and $P_{\pm1}'=(X\pm1)^9$ have their only roots at
$0$ and $\mp1$, respectively, where the corresponding polynomial equals
$1$; hence $P_0,P_{+1}$, and $P_{-1}$ are squarefree.  Using
$P_{\pm1}=X^{10}\pm X^9\pm X-1$, a direct calculation also gives
\[
  \gcd(P_0,P_{+1})=\gcd(P_0,P_{-1})=\gcd(P_{+1},P_{-1})=1.
\]
Consequently, every nonempty product of the polynomials $P_0,P_{+1},P_{-1}$ is
squarefree.

For
$\boldsymbol\varepsilon=(\varepsilon_0,\varepsilon_{+1},\varepsilon_{-1})
\in\{\pm1\}^3$, let
\[
  N_{\boldsymbol\varepsilon}
  :=\#\{c\in\Fq:\chi(P_j(c))=\varepsilon_j
                  \text{ for }j\in\{0,+1,-1\}\}.
\]
Put $J:=\{0,+1,-1\}$.  Since $-1$ is a nonsquare in $\Fq$, one has
$P_j(c)\ne0$ for every $c\in\Fq$ and $j\in J$.
Hence the indicator of the condition $\chi(P_j(c))=\varepsilon_j$ is
\[
  \mathbf 1_{\{\chi(P_j(c))=\varepsilon_j\}}
  =\frac{1+\varepsilon_j\chi(P_j(c))}{2}.
\]
Consequently,
\[
\begin{aligned}
  N_{\boldsymbol\varepsilon}
  &=\frac18\sum_{c\in\Fq}
    \prod_{j\in J}\bigl(1+\varepsilon_j\chi(P_j(c))\bigr)\\
  &=\frac18\Bigg(
    q+\sum_{j\in J}\varepsilon_j\sum_{c\in\Fq}\chi(P_j(c))\\
  &\qquad
    +\sum_{\substack{\{j,k\}\subseteq J\\ j\ne k}}
       \varepsilon_j\varepsilon_k
       \sum_{c\in\Fq}\chi(P_j(c)P_k(c))\\
  &\qquad
    +\varepsilon_0\varepsilon_{+1}\varepsilon_{-1}
       \sum_{c\in\Fq}\chi(P_0(c)P_{+1}(c)P_{-1}(c))
    \Bigg).
\end{aligned}
\]
The products involving one, two, or three of the $P_j$ have degrees
$10$, $20$, or $30$, respectively.  Applying the Weil bound for
quadratic-character sums~\cite[Lemma~1]{LangeWinterhof02} to these seven sums
therefore gives
\[
  N_{\boldsymbol\varepsilon}
  \ge \frac{q-(3\cdot9+3\cdot19+29)\sqrt q}{8}
  =\frac{q-113\sqrt q}{8}.
\]
Since $q\ge3^9$, the last lower bound is
greater than $3$.  Thus every sign pattern occurs for some
$c\in\Fq\setminus\Fthree$.

Since $g(\Fq^*)=(\Fq^*)^2$, the condition $\chi(P_0(c))=-1$ is equivalent
to $c\in\mathcal C_g$ for $c\in\Fq\setminus\Fthree$.  Consequently, each
of the three patterns
\[
  (-1,-1,-1),\qquad
  (-1,+1,-1),\qquad
  (-1,+1,+1)
\]
occurs for some parameter $c\in\mathcal C_g$.  Moreover,
by~\eqref{eq:kappa-c},
\[
  \kappa_c
  =\#\{j\in\{+1,-1\}:\chi(P_j(c))=+1\}.
\]
Therefore
\[
  (-1,-1,-1)\Longrightarrow\kappa_c=0,\qquad
  (-1,+1,-1)\Longrightarrow\kappa_c=1,\qquad
  (-1,+1,+1)\Longrightarrow\kappa_c=2.
\]
Hence
\[
  \{\kappa_c:c\in\mathcal C_g\}=\{0,1,2\}.
\]
These parameters realize all three spectra in
Corollary~\ref{cor:lc-bct-spectrum}\textup{(a)}.
Likewise, each of the patterns
\[
  (+1,-1,-1),\qquad
  (+1,+1,-1),\qquad
  (+1,+1,+1)
\]
occurs for some $c\in\Fq\setminus(\Fthree\cup\mathcal C_g)$ and gives
$\kappa_c=0,1,2$, respectively.  Thus all three spectra in
Corollary~\ref{cor:lc-bct-spectrum}\textup{(b)} also occur.
The parameters $c=0$ and $c=\pm1$ give the two remaining spectra in
parts \textup{(c)} and \textup{(d)}, respectively.
Consequently, for every odd $n\ge9$, the APN family
$\{\widetilde G_c:c\in\Fq\}$ arising from $g(x)=x^{10}$ realizes all
eight boomerang spectra in Corollary~\ref{cor:lc-bct-spectrum}.
Subsection~\ref{sec:known-examples} presents further examples from known DO PN
families, verified by exhaustive computation.
\end{remark}

\subsection{Applications to known DO PN families}
\label{sec:known-examples}

The preceding affine-perturbation results
apply to every DO PN function over $\F_q$, where $q=3^n$ and
$n>1$ is odd.  For orientation, Table~\ref{tab:known-pn-q3n} records selected known
PN families over $\F_q$ that remain valid in this setting.  The rows are arranged in the
chronological order of the cited constructions and are assigned the local labels
$f_1,\ldots,f_5$.  The Coulter--Matthews monomial
$f_2$ is included for context but is not represented by a DO polynomial in
general and is not used below.

\begin{table}[!htbp]
\centering
\small
\renewcommand{\arraystretch}{1.22}
\begin{tabularx}{\textwidth}{@{}cY Y c@{}}
\toprule
Family & PN function over $\F_q$, $q=3^n$ & Conditions for validity & Source \\
\midrule
$f_1$ & $x^{3^k+1}$ &
$k\ge0$.  The general condition $2\nmid n/\gcd(n,k)$ is automatic because $n$ is odd. &
\cite{CM97,DO68} \\
\addlinespace
$f_2$ & $x^{(3^k+1)/2}$ &
$k$ odd and $\gcd(n,k)=1$.  This is the Coulter--Matthews family; it is not
represented by a DO polynomial in general. &
\cite{CM97} \\
\addlinespace
$f_3$ & $x^{10}-\lambda x^6-\lambda^2x^2$ &
$\lambda\in\F_q^*$.  The original odd-degree condition is automatic because $n$ is odd. &
\cite{DY06} \\
\addlinespace
$f_4$ &
$\displaystyle x^{3^k+1}-\omega^{3^s-1}x^{3^s+3^{2s+k}}$ &
Here $n=3s$, $3\nmid s$, $\operatorname{ord}(\omega)=3^n-1$, and
$s\equiv k\pmod3$.  The further condition
$s/\gcd(s,k)$ odd is automatic here because $n=3s$ is odd. &
\cite[Thm.~1]{ZKW09} \\
\addlinespace
$f_5$ &
$\displaystyle x^{3^t+1}-\mu x^{3^{2s}+3^{s+t}}$ &
Here $n=3s$, $\operatorname{ord}(\mu)=3^{2s}+3^s+1$, and
$\dfrac{s+t}{\gcd(s,t)}\equiv0\pmod3$.  The further condition
$\dfrac{s}{\gcd(s,t)}$ odd is automatic here because $n=3s$ is odd. &
\cite[Thm.~4]{Bier10} \\
\bottomrule
\end{tabularx}
\caption{Selected known PN families over $\F_q$ for $q=3^n$ with $n$ odd.}
\label{tab:known-pn-q3n}
\end{table}
\FloatBarrier

\paragraph{Relation between $f_4$ and $f_5$.}
For a member of the Zha--Kyureghyan--Wang $f_4$ family, put $t=s+k$ and
$\mu=\omega^{1-3^s}$.  The corresponding $f_5$ satisfies Bierbrauer's conditions and
\[
  f_4\!\left(x^{3^{2s}}\right)=-\omega^{3^s-1}f_5(x).
\]
Hence $f_4$ and the corresponding $f_5$ are linearly equivalent, and therefore
EA- and CCZ-equivalent; this is the
relation underlying Bierbrauer's description of Theorem~4 as a slight generalization of the
Zha--Kyureghyan--Wang family~\cite{Bier10}.

For $i\in\{1,3,4,5\}$, let $g_i=f_i$ be any valid member of the
corresponding DO PN family in Table~\ref{tab:known-pn-q3n}, and set
\[
\begin{aligned}
  \tau_i&:=g_i(1),\qquad
  G_i(x):=g_i(x)+\tau_i\epsilon(x),\\
  \mathcal C_i
    &:=\{c\in\Fq\setminus\Fthree:
          g_i(c)+\tau_i\notin g_i(\Fq)\},\\
  \widetilde G_{i,c}(x)
    &:=G_i(x)+D_cg_i(x)
      =g_i(x+c)+\tau_i\epsilon(x)
      \qquad(c\in\Fq).
\end{aligned}
\]
Since $g_i$ is DO PN, Theorem~\ref{thm:CM-two-to-one} gives
$\tau_i\ne0$.  For $c\in\Fq\setminus\Fthree$, let
$\kappa_{i,c}$ denote the quantity in \eqref{eq:kappa-c} with
$g=g_i$ and $\tau=\tau_i$.

Applying Corollary~\ref{cor:natural-derivative-perturbations} to $g_i$
gives
\[
  |\mathcal C_i|=\frac{q-3}{2},
  \qquad
  \delta_{\widetilde G_{i,c}}=2\quad(c\in\Fq),
  \qquad
  \beta_{\widetilde G_{i,c}}=1
  \ \Longleftrightarrow\ c\in\mathcal C_i.
\]
Moreover, $D_cg_i$ is affine, so
$\widetilde G_{i,c}\sim_{\EA}G_i$.  Hence
Lemma~\ref{lem:ccz-global-ddt-spectrum} and
Theorem~\ref{thm:diff-spectrum} give its differential spectrum, whereas
Corollary~\ref{cor:lc-bct-spectrum}\textup{(a)} directly gives its boomerang
spectrum.  Consequently, for every $c\in\mathcal C_i$,
\[
\begin{aligned}
  \mathcal S_{\DDT}(\widetilde G_{i,c})
    &=\{\!\{\,0^{4(q-3)},\,
          1^{\,2q+(q-3)(q-8)},\,
          2^{4(q-3)}\,\}\!\},\\
  \mathcal S_{\BCT}(\widetilde G_{i,c})
    &=\{\!\{\,0^{(q-5)^2+8\kappa_{i,c}},\,
          1^{\,8(q-3)-8\kappa_{i,c}}\,\}\!\}.
\end{aligned}
\]

\paragraph{Computational evidence for all eight spectrum types.}
Table~\ref{tab:kappa-evidence} reports computations for selected members of
the $f_1$, $f_3$, and $f_5$ families over $\F_{3^7}$, $\F_{3^7}$, and
$\F_{3^{15}}$, respectively.  Each row contains two triples, whose entries
correspond, in order, to $\kappa_{i,c}=0,1,2$.  The first triple counts the
parameters $c\in\mathcal C_i$, and the second counts those
$c\in\Fq\setminus(\Fthree\cup\mathcal C_i)$.
These distributions were obtained by exhaustive computation for the functions
and parameters listed in the table.

For the $f_5$ row, $n=15$ is the smallest admissible extension degree in the
present odd-degree setting for which the family yields a genuine binomial.  At
the smaller admissible degrees $n=3$ and $n=9$, the two monomial terms have
congruent exponents modulo $3^n-1$, so every admissible $f_5$ member reduces to
a nonzero scalar multiple of a Gold monomial.  For this example, take
$\F_{3^{15}}=\Fthree[X]/(X^{15}+2X^2+1)$,
$\alpha=X\pmod{X^{15}+2X^2+1}$, and
$f_5(x)=x^4-\alpha^{242}x^{3^{10}+3^6}$.  The defining polynomial is
primitive, and hence
$\operatorname{ord}(\alpha^{242})=(3^{15}-1)/(3^5-1)=3^{10}+3^5+1$.
Together with $s=5$ and $t=1$, this verifies the defining conditions of the
Bierbrauer $f_5$ family, so this function is an admissible member of that
family and is a genuine binomial.

\begin{table}[!htbp]
\centering
\footnotesize
\renewcommand{\arraystretch}{1.18}
\setlength{\tabcolsep}{3.5pt}
\begin{tabularx}{\textwidth}{@{}c c Y c c@{}}
\toprule
Family & $n$ & Parameters
  & $c\in\mathcal C_i$
  & $c\in\Fq\setminus(\Fthree\cup\mathcal C_i)$ \\
\midrule
$f_1$ & $7$ & $k=2$
  & $(252,\,588,\,252)$
  & $(294,\,504,\,294)$ \\
$f_3$ & $7$ & $\lambda=1$
  & $(294,\,504,\,294)$
  & $(252,\,588,\,252)$ \\
$f_5$ & $15$ & $s=5$, $t=1$, $\mu=\alpha^{242}$
  & $(1\,785\,492,\,3\,600\,988,\,1\,787\,972)$
  & $(1\,800\,494,\,3\,575\,944,\,1\,798\,014)$ \\
\bottomrule
\end{tabularx}
\caption{Computed distributions of $\kappa_{i,c}$ for selected members of
the $f_1$, $f_3$, and $f_5$ families.  Each triple lists, in order, the
numbers of parameters with $\kappa_{i,c}=0,1,2$.}
\label{tab:kappa-evidence}
\end{table}
\FloatBarrier
\Needspace{8\baselineskip}

All six counts in each row of Table~\ref{tab:kappa-evidence} are positive.
Thus, for each displayed PN function, the computations realize all three
boomerang-uniformity-one spectrum types in
Corollary~\ref{cor:lc-bct-spectrum}\textup{(a)} and all three
boomerang-uniformity-two spectrum types in
Corollary~\ref{cor:lc-bct-spectrum}\textup{(b)}.  The two remaining
boomerang-uniformity-two spectrum types, corresponding to $c=0$ and
$c\in\Fthree^*$, are given by
Corollary~\ref{cor:lc-bct-spectrum}\textup{(c)} and \textup{(d)}, respectively.
Consequently, the computations, together with these two special cases,
provide numerical evidence that all eight spectrum types described in
Corollary~\ref{cor:lc-bct-spectrum} are realized for each of the three
displayed PN functions.

\begin{remark}\label{rem:known-families-degree}
For every $c\in\Fq$, the function $\widetilde G_{i,c}$ has algebraic degree $2n$.
Indeed, $\tau_i\ne0$, so $\tau_i\epsilon$ has a nonzero $x^{q-1}$ term of
base-$3$ weight $2n$, whereas $g_i(x+c)$ has algebraic degree exactly $2$,
since $g_i$ is a DO PN function and translation preserves algebraic degree.
For functions of algebraic degree at least $2$, algebraic degree is invariant
under EA equivalence; hence this common degree cannot distinguish the
EA-equivalence classes of these functions.
\end{remark}

For arbitrary DO PN functions $g,h$, let $G,H$ be their
sign-switches and let $\widetilde G,\widetilde H$ be associated functions
with boomerang uniformity one.  Corollary~\ref{cor:ccz-inequivalent-associated-functions}
proves that
\[
  g\not\sim_{\CCZ}h
  \quad\Longrightarrow\quad
  G\not\sim_{\CCZ}H
  \quad\text{and}\quad
  \widetilde G\not\sim_{\CCZ}\widetilde H.
\]
Since each $\widetilde G_{i,c}$ is EA-equivalent to $G_i$, it follows that if
two selected PN functions $g_i,g_j$ are CCZ-inequivalent, then every pair
$\widetilde G_{i,c},\widetilde G_{j,c'}$, with $c\in\mathcal C_i$ and
$c'\in\mathcal C_j$, is CCZ-inequivalent.
Theorem~\ref{thm:three-beta-one-classes} uses the orders of the nuclei of the
associated presemifields to show that,
for infinitely many odd extension degrees $n$, the Gold $f_1$, Ding--Yuan
$f_3$, and Bierbrauer $f_5$ PN functions over $\F_{3^n}$ are pairwise
CCZ-inequivalent; consequently, their associated APN functions with boomerang
uniformity one are also pairwise CCZ-inequivalent.

\section{Switch-rigidity and CCZ-inequivalence}
\label{sec:switch-rigidity-inequivalence}

This section establishes the rigidity results needed to prove that
CCZ-inequivalent PN inputs yield CCZ-inequivalent associated functions.

\subsection{Uniqueness of the vertical subspace and CCZ-to-EA collapse}\label{sec:unique-vertical-subspace}

We now turn to CCZ equivalence, using the graph notation and equivalence
conventions introduced in Section~\ref{sec:preliminaries}.
Let $\pi(a,b):=a$ be the projection to the first coordinate, and put
$\mathcal K:=\ker\pi=\{0\}\times\Fq$.  For every function $F:\Fq\to\Fq$ and every
$v=(0,b)\in\mathcal K\setminus\{(0,0)\}$, one has $\delta_F(v)=0$ because
$b\ne0$.

For orientation, recall the standard fact that, for arbitrary PN functions over finite
fields of odd characteristic, CCZ equivalence coincides with EA equivalence.  This fact
does not apply directly here, because the sign-switched functions considered in this paper
are APN and therefore are not PN.  The point of Lemma~\ref{lem:ccz-to-ea} below is to prove the
corresponding collapse for the present switched maps, using the uniqueness of
the vertical subspace established in Lemma~\ref{lem:unique-vertical-subspace}.

\begin{lemma}\label{lem:unique-vertical-subspace}
Let $G$ be the sign-switch of a DO PN function $g$, as in
\eqref{eq:sign-switch}.  Then $\mathcal K=\ker\pi$ is the unique
$n$-dimensional $\Fthree$-subspace $U\le \Fq\times\Fq$ such that
\begin{equation}\label{eq:kernel-zero-subspace}
  U\setminus\{(0,0)\}
  \subseteq
  \{v\in\Fq\times\Fq:\delta_G(v)=0\}.
\end{equation}
\end{lemma}

\begin{proof}
Let $U\le\Fq\times\Fq$ be an $n$-dimensional $\Fthree$-subspace satisfying
\eqref{eq:kernel-zero-subspace}.

If $\pi(U)=\{0\}$, then $U\subseteq\mathcal K$.  Since both $U$ and
$\mathcal K$ have $\Fthree$-dimension $n$, this gives $U=\mathcal K$.

Assume, for contradiction, that $\pi(U)\ne\{0\}$.  Put
\[
  W:=\{b\in\Fq:(0,b)\in U\},
  \qquad
  U_0:=\{0\}\times W,
  \qquad
  r:=\dim_{\Fthree}W.
\]
For each $a\in\pi(U)$, define the fiber of the restricted projection by
\[
  U_a:=(\pi|_U)^{-1}(a)
      =U\cap\bigl(\{a\}\times\Fq\bigr).
\]
Then one has the disjoint union
\[
  U=\bigsqcup_{a\in\pi(U)}U_a.
\]
If $(a,b)\in U_a$, then
\[
  U_a=(a,b)+U_0=\{a\}\times(b+W),
  \qquad
  |U_a|=3^r.
\]

Choose $a\in\pi(U)\setminus\{0\}$ and $(a,b)\in U_a$.  For every
$c\in b+W$, the vector $(a,c)$ is a nonzero element of $U$, and hence
\[
  \delta_G(a,c)=0.
\]
We cannot have $a\in\Fthreestar$, because $D_aG$ is then a permutation by
Theorem~\ref{thm:diff-spectrum}.  Thus $a\in\Fqstar\setminus\Fthree$.
Corollary~\ref{cor:exceptional-values} and
\eqref{eq:exceptional-derivative-value-sets} give
\[
  b+W\subseteq\{\pm A,\pm B\},
\]
where the four elements are distinct and nonzero.  Hence
$3^r=|b+W|\le4$, so $r\in\{0,1\}$.

Suppose that $r=1$.  Choose $0\ne t\in W$.  Then
\[
  W=\{0,t,-t\},
  \qquad
  b+W=\{b,b+t,b-t\},
\]
whose elements sum to zero.  On the other hand, every three-element subset
of $\{\pm A,\pm B\}$ has sum equal to the negative of its omitted element,
and hence has nonzero sum.  This is impossible.  Therefore $r=0$.

It follows that every $U_a$ is a singleton.  Since the preceding disjoint
union has $|U|=3^n=q$, one has $|\pi(U)|=q$, and hence $\pi(U)=\Fq$.
Choose $b\in\Fq$ such that $(1,b)\in U$.  Then
\eqref{eq:kernel-zero-subspace} gives $\delta_G(1,b)=0$, whereas $D_1G$ is
a permutation and therefore $\delta_G(1,b)=1$, a contradiction.

Hence $\pi(U)=\{0\}$, and therefore $U=\mathcal K$.
\end{proof}

\begin{lemma}[CCZ collapses to EA for switched maps]\label{lem:ccz-to-ea}
Let $G$ and $H$ be sign-switches of DO PN functions. If
\[
  G\sim_{\CCZ}H,
\]
then the linear part of every affine CCZ map sending $\mathcal G_G$ to
$\mathcal G_H$ preserves the kernel $\mathcal K=\ker\pi$ of the projection to
the first coordinate. Consequently, $G$ and $H$ are EA-equivalent.
\end{lemma}

\begin{proof}
Let
\[
  \mathcal A(P)=\mathcal M(P)+\mathbf c
\]
be an affine automorphism of $\Fq\times\Fq$ satisfying
\[
  \mathcal A(\mathcal G_G)=\mathcal G_H.
\]
By Lemma~\ref{lem:ddt-entry-transport}, if $\delta_G(v)=0$, then
$\delta_H(\mathcal Mv)=0$.  Since every
$v\in\mathcal K\setminus\{(0,0)\}$ satisfies $\delta_G(v)=0$, every
$u\in\mathcal M(\mathcal K)\setminus\{(0,0)\}$ satisfies $\delta_H(u)=0$.
Thus $\mathcal M(\mathcal K)$ is an $n$-dimensional subspace satisfying
\eqref{eq:kernel-zero-subspace}, with $H$ in place of $G$.  By
Lemma~\ref{lem:unique-vertical-subspace},
\[
  \mathcal M(\mathcal K)=\mathcal K.
\]

For $x,y\in\Fq$, the $\Fthree$-linearity of $\mathcal M$ gives $\mathcal M(x,y)=\mathcal M(x,0)+\mathcal M(0,y)$.
Since $(0,y)\in\mathcal K$ and
$\mathcal M(\mathcal K)=\mathcal K$, we have
$\mathcal M(0,y)\in\mathcal K$.  Hence the first coordinate of
$\mathcal M(x,y)$ depends only on $x$.  Thus there exist invertible
$\Fthree$-linear maps
\[
  M_1,M_2:\Fq\longrightarrow\Fq
\]
and an $\Fthree$-linear map
\[
  B:\Fq\longrightarrow\Fq
\]
such that, writing $\mathbf c=(a_0,b_0)$ and using
$\mathcal A(\mathcal G_G)=\mathcal G_H$,
\begin{equation}\label{eq:ccz-ea-normal-form}
\begin{aligned}
  \mathcal M(x,y)&=(M_1x,M_2y+B(x)),\\
  \mathcal A(x,y)&=\mathcal M(x,y)+\mathbf c
  =\bigl(M_1x+a_0,\,M_2y+B(x)+b_0\bigr),\\
  H(M_1x+a_0)&=M_2G(x)+B(x)+b_0
  \qquad(x\in\Fq).
\end{aligned}
\end{equation}
This is precisely EA equivalence.
\end{proof}

\begin{lemma}\label{lem:input-map-preserves-fthree}
In the notation of Lemma~\ref{lem:ccz-to-ea}, the input linear map $M_1$ satisfies
\begin{equation}\label{eq:input-map-preserves-fthree}
  M_1(\Fthreestar)=\Fthreestar.
\end{equation}
\end{lemma}

\begin{proof}
Taking the derivative of the last identity in
\eqref{eq:ccz-ea-normal-form} in direction $a$, we obtain
\[
  D_{M_1a}H(M_1x+a_0)=M_2D_aG(x)+B(a).
\]
Therefore $D_{M_1a}H$ is a permutation if and only if $D_aG$ is a permutation.
By Theorem~\ref{thm:diff-spectrum}, for either $F=G$ or $F=H$, the derivative
$D_aF$ is a permutation exactly when $a\in\Fthreestar$.  Hence
\eqref{eq:input-map-preserves-fthree} follows.
\end{proof}

\subsection{Uniqueness of the underlying PN function and switch-rigidity}\label{sec:underlying-pn-uniqueness}

The following lemma identifies the original PN function as the unique PN
function that agrees with its sign-switch outside a two-point translate of
$\Fthreestar$.

\begin{lemma}\label{lem:underlying-pn-uniqueness}
Let $h:\Fq\to\Fq$ be a DO PN function, put
$\sigma:=h(1)$, and let $H(x):=h(x)+\sigma\epsilon(x)$ be its sign-switch,
with $\epsilon$ as in \eqref{eq:sign-switch}.
Suppose $f:\Fq\to\Fq$ is PN and agrees with $H$ outside a two-point set
$C=c+\Fthreestar=\{c+1,c-1\}$.  Then $c=0$ and $f=h$.
\end{lemma}

\begin{proof}
We proceed in three steps.

For this proof, write
$L_a(x):=h(x+a)-h(x)-h(a)$ and, for
$a\in\Fqstar\setminus\Fthree$, put $u_a:=L_a^{-1}(-\sigma)$.
For every $a\in\Fqstar\setminus\Fthree$,
Corollary~\ref{cor:derivative-two-element-fibers}, applied with
$(g,G,\tau)$ replaced by $(h,H,\sigma)$, shows that the only
nonsingleton fibers of $D_aH$ are
\begin{equation}\label{eq:underlying-pn-fibers}
\begin{gathered}
  \{1,1+u_a\},\qquad \{-1,-1+u_a\},\\
  \{1-a,1-a-u_a\},\qquad \{-1-a,-1-a-u_a\}.
\end{gathered}
\end{equation}
For every $a\in\Fqstar\setminus\Fthree$, since $f=H$ on
$\Fq\setminus C$, one has $D_af(x)=D_aH(x)$ for every
$x\in\Fq\setminus\bigl(C\cup(C-a)\bigr)$.
Whenever $|C\cup(C-a)|=4$, the set $C\cup(C-a)$ meets each of the four
two-element fibers in \eqref{eq:underlying-pn-fibers} exactly once.  Indeed,
if one of these fibers were disjoint from $C\cup(C-a)$, its two points
would still have the same image under $D_af$, contradicting that $D_af$
is a permutation.  Since the four fibers are pairwise disjoint and
$C\cup(C-a)$ has four elements, each fiber is met exactly once.

\Needspace{6\baselineskip}
\medskip
\noindent\textbf{Step 1: $c\in\Fqstar\setminus\Fthree$ is impossible.}
\par\smallskip
Suppose that $c\in\Fqstar\setminus\Fthree$ and consider the derivative in
direction $c$.  Since $f$ differs from $H$ only on
$C=c+\Fthreestar$, the derivatives $D_cf$ and $D_cH$ agree outside the
four-element set
\[
  C\cup(C-c)=\{1,-1,1+c,-1+c\}.
\]
The preceding observation applies.  The points $1$ and $-1$ meet the first
two fibers in \eqref{eq:underlying-pn-fibers}, so $1+c$ and $-1+c$ meet the
last two.  Neither is one of their first listed points $1-c$ and $-1-c$:
any such equality would give
$2c\in\Fthree$, contrary to $c\notin\Fthree$.  Hence
\[
  \{1+c,-1+c\}=\{1-c-u_c,-1-c-u_c\}.
\]
Comparing sums gives $u_c=c$.  Therefore
\[
  -\sigma=L_c(u_c)=L_c(c)=h(2c)-2h(c)=-h(c),
\]
because $h(2c)=h(-c)=h(c)$.  Hence
$h(c)=\sigma=h(1)$, which is impossible by
Corollary~\ref{cor:special-preimages} because $c\notin\Fthree$.

\Needspace{6\baselineskip}
\medskip
\noindent\textbf{Step 2: $c\in\Fthreestar$ is impossible.}
\par\smallskip
Let $c\in\Fthreestar$ and choose any
$a\in\Fqstar\setminus\Fthree$.  Since
$C=c+\Fthreestar=\{0,-c\}$, the derivatives $D_af$ and $D_aH$ agree
outside the four-element set
\[
  C\cup(C-a)=\{0,-c,-a,-c-a\}.
\]
The preceding observation applies.  The points $-c$ and $-c-a$ are the
first listed points of two fibers in \eqref{eq:underlying-pn-fibers}.
The remaining points $0$ and $-a$ therefore meet the other two fibers.
Their first listed points are $c$ and $c-a$, and neither $0$ nor $-a$
equals either of them, because $c\ne0$ and $a\notin\Fthree$.  Thus
\[
  \{0,-a\}=\{c+u_a,c-a-u_a\}.
\]
Comparing sums gives $-a=2c-a$, and hence $2c=0$, contrary to
$c\in\Fthreestar$.

\Needspace{6\baselineskip}
\medskip
\noindent\textbf{Step 3: $c=0$, and it remains to show $f=h$.}
\par\smallskip
By Steps~1 and~2, the only remaining possibility is $c=0$.  Now
$C=\Fthreestar=\{1,-1\}$.  Since $f=H=h$ outside $C$, the permutations
$D_1f$ and $D_1h$ agree outside $C\cup(C-1)=\Fthree$.  Their image sets on $\Fthree$ are
therefore equal, each being the complement of their common image of
$\Fq\setminus\Fthree$.  Evaluating at $0,1,-1$ gives
\[
  \{f(1),f(-1)-f(1),-f(-1)\}
  =D_1f(\Fthree)=D_1h(\Fthree)
  =
  \{\sigma,0,-\sigma\}.
\]
Choose any $a\in\Fqstar\setminus\Fthree$ and put
\[
  A:=D_ah(1)=h(a+1)-\sigma,
  \qquad
  B:=D_ah(-1)=h(a-1)-\sigma.
\]
Since $f=h$ outside $C=\Fthreestar$, equation
\eqref{eq:exceptional-input-set} shows that $D_af$ and $D_ah$
agree outside $E_a$.  Both derivatives are permutations, so their images of
$E_a$ coincide.  Therefore \eqref{eq:exceptional-derivative-value-sets} and
\eqref{eq:four-exclusions-general}, applied with $(g,G,\tau)$ replaced by
$(h,H,\sigma)$, give
\[
  D_af(E_a)=D_ah(E_a)=\{\pm A,\pm B\},
  \qquad
  A,B,A+B,A-B\notin\sigma\Fthree.
\]
In particular, $D_af(\pm1)\in\{\pm A,\pm B\}$.  Suppose first that
$f(1)=0$.  Since $a+1\notin C$, one has $f(a+1)=h(a+1)$, and hence
\[
  D_af(1)=f(a+1)-f(1)=h(a+1)=A+\sigma\in\{\pm A,\pm B\}.
\]
Equality with one of $\pm A,\pm B$ would force $\sigma=0$ or one of
$A,A-B,A+B$ to lie in $\sigma\Fthree$, a contradiction.  Thus
$f(1)\ne0$.  If $f(-1)=0$, then $a-1\notin C$ similarly gives
$D_af(-1)=f(a-1)-f(-1)=h(a-1)=B+\sigma\in\{\pm A,\pm B\}$.
The same check, with $B$ in place of $A$, gives a contradiction.  Thus
$f(-1)\ne0$.

Consequently, the equality
$\{f(1),f(-1)-f(1),-f(-1)\}=\{\sigma,0,-\sigma\}$ forces
$f(-1)-f(1)=0$, so $f(1)=f(-1)\in\{\pm\sigma\}$.  If
$f(1)=f(-1)=-\sigma$, then $f=H$, because $H(\pm1)=-\sigma$;
this contradicts
$\delta_f=1$ and $\delta_H=2$ from Theorem~\ref{thm:diff-spectrum}.
Therefore $f(1)=f(-1)=\sigma=h(1)=h(-1)$.  Together with $f=h$
outside $C$, this proves $f=h$.
\end{proof}

\begin{remark}
\label{rem:pn-uniqueness-scope}
In Lemma~\ref{lem:underlying-pn-uniqueness}, $f$ is assumed only to be PN;
its quadraticity and representation by a DO polynomial follow from the
conclusion $f=h$.
\end{remark}

\begin{theorem}[Switch-rigidity]\label{thm:switch-rigidity}
Let $g,h:\Fq\to\Fq$ be DO PN functions, and let $G,H$ be
their sign-switches.  If $G\sim_{\CCZ}H$, then $g\sim_{\EA}h$.
\end{theorem}

\begin{proof}
Assume $G\sim_{\CCZ}H$.  By Lemma~\ref{lem:ccz-to-ea}, the CCZ
equivalence is induced by an EA map $\mathcal A$ of the normal form
\eqref{eq:ccz-ea-normal-form}, and \eqref{eq:input-map-preserves-fthree} gives
$M_1(\Fthreestar)=\Fthreestar$.  By \eqref{eq:sign-switch}, the graph
$\mathcal G_g$ differs from $\mathcal G_G$ exactly at the two inputs
$\Fthreestar=\{1,2\}$.  Since the first coordinate of $\mathcal A$ is
$M_1x+a_0$, the graph $\mathcal A(\mathcal G_g)$ differs from
$\mathcal A(\mathcal G_G)=\mathcal G_H$ exactly at
$a_0+M_1(\Fthreestar)=a_0+\Fthreestar$.

Because $g$ is PN and $\mathcal A$ is an EA map,
$\mathcal A(\mathcal G_g)$ is the graph $\mathcal G_f$ of a PN function
$f$.  Thus $f$ agrees with $H$ outside the two-point set
$a_0+\Fthreestar$.  Lemma~\ref{lem:underlying-pn-uniqueness}, applied to $H$,
gives $a_0=0$ and $f=h$.  Therefore
$\mathcal A(\mathcal G_g)=\mathcal G_h$, so $g\sim_{\EA}h$.
\end{proof}

\begin{corollary}
\label{cor:linear-equivalence-sign-switches}
The EA equivalence between $g$ and $h$ in
Theorem~\ref{thm:switch-rigidity} is in fact a linear equivalence;
that is, there exist invertible
$\Fthree$-linear maps $M_1,M_2$ such that
\[
  H(M_1x)=M_2G(x),
  \qquad
  h(M_1x)=M_2g(x)
  \qquad(x\in\Fq).
\]
More precisely, every affine CCZ map $\mathcal A$ sending $\mathcal G_G$
to $\mathcal G_H$ has the form
\[
  \mathcal A(x,y)=(M_1x,M_2y)
\]
and also satisfies
\[
  \mathcal A(\mathcal G_g)=\mathcal G_h.
\]
\end{corollary}

\begin{proof}
Let $\mathcal A$ be any affine CCZ map sending $\mathcal G_G$ to
$\mathcal G_H$.  By Lemma~\ref{lem:ccz-to-ea}, it has the normal form
\eqref{eq:ccz-ea-normal-form}.  The proof of
Theorem~\ref{thm:switch-rigidity}, applied to this map, gives $a_0=0$ and
$\mathcal A(\mathcal G_g)=\mathcal G_h$.
Evaluating \eqref{eq:ccz-ea-normal-form} at $x=0$ gives $b_0=0$, since
$G(0)=H(0)=0$.  Since $G$ and $H$ are even, comparing the identity
\[
  H(M_1x)=M_2G(x)+B(x),
\]
obtained from \eqref{eq:ccz-ea-normal-form}, at $x$ and $-x$ gives
$2B(x)=0$, and hence $B=0$.  Thus $\mathcal A(x,y)=(M_1x,M_2y)$.
The identities $H(M_1x)=M_2G(x)$ and $h(M_1x)=M_2g(x)$ now follow from
the corresponding graph equalities.
\end{proof}

\subsection{CCZ-inequivalence arising from inequivalent PN inputs}
\label{sec:ccz-inequivalence-families}

Throughout this subsection, $g$ denotes a DO PN function, $G$
denotes its sign-switch as in \eqref{eq:sign-switch}, and $\widetilde G$
denotes any one of the functions $\widetilde G_c$ in
\eqref{eq:main-family}, with $c\in\mathcal C_g$.  We suppress $c$ because
its choice plays no role below.  By
Corollary~\ref{cor:natural-derivative-perturbations},
\begin{equation}\label{eq:sec8-standing}
  \delta_G=\delta_{\widetilde G}=2,
  \qquad
  \beta_G=2,
  \qquad
  \beta_{\widetilde G}=1,
  \qquad
  \widetilde G\sim_{\EA}G.
\end{equation}
For a second DO PN function $h$, the notation
$H,\widetilde H$ is used analogously, and the corresponding relations in
\eqref{eq:sec8-standing} hold.

\begin{corollary}
\label{cor:ccz-inequivalent-associated-functions}
Let $g,h:\Fq\to\Fq$ be DO PN functions, with the notation
above.  If
\[
  g\not\sim_{\CCZ}h,
\]
then both
\[
  G\not\sim_{\CCZ}H
  \qquad\text{and}\qquad
  \widetilde G\not\sim_{\CCZ}\widetilde H.
\]
\end{corollary}

\begin{proof}
Indeed, if $G\sim_{\CCZ}H$, then Theorem~\ref{thm:switch-rigidity} gives
$g\sim_{\EA}h$, and hence $g\sim_{\CCZ}h$, contrary to the hypothesis.
Thus $G\not\sim_{\CCZ}H$.  Moreover, since EA equivalence implies CCZ
equivalence, \eqref{eq:sec8-standing} yields
$\widetilde G\sim_{\CCZ}\widetilde H$ if and only if
$G\sim_{\CCZ}H$; hence $\widetilde G\not\sim_{\CCZ}\widetilde H$.
\end{proof}

\begin{theorem}[Three CCZ-inequivalent classes with boomerang uniformity one]
\label{thm:three-beta-one-classes}
Let $q=3^n=3^{3s}$, where $s>1$ is odd,
choose $1\le t<3s$, put $\ell:=\gcd(s,t)$, and assume
\[
  \ell>1,
  \qquad
  s\nmid t,
  \qquad
  \frac{s+t}{\ell}\equiv0\pmod3.
\]
Let $\omega$ be a primitive element of $\F_{3^{3s}}$, and define
\[
\begin{aligned}
  g_1(x)&:=x^{3^s+1},\\
  g_3(x)&:=x^{10}-x^6-x^2,\\
  g_5(x)&:=x^{3^t+1}
     -\omega^{3^s-1}x^{3^{2s}+3^{s+t}}.
\end{aligned}
\]
For $i\in\{1,3,5\}$, let $\widetilde G_i$ be any one of the functions in
\eqref{eq:main-family}, with $g$ replaced by $g_i$ and with its parameter
in $\mathcal C_{g_i}$.  Then
$\widetilde G_1,\widetilde G_3$, and $\widetilde G_5$ are pairwise
CCZ-inequivalent APN functions, and
\[
  \beta_{\widetilde G_1}
  =\beta_{\widetilde G_3}
  =\beta_{\widetilde G_5}
  =1.
\]
Taking $s=5\ell$ and $t=\ell$ for any odd $\ell>1$ yields infinitely many odd
extension degrees $n=15\ell$ for which such a triple exists; the first is
$n=45$.
\end{theorem}

\begin{proof}
Equation~\eqref{eq:sec8-standing}, applied to each $g_i$, gives the APN and
boomerang-uniformity assertions.  It remains to prove that $g_1,g_3$, and
$g_5$ are pairwise CCZ-inequivalent.
Since $\omega$ is primitive, one has
$\operatorname{ord}(\omega^{3^s-1})=(3^{3s}-1)/(3^s-1)=3^{2s}+3^s+1$.
Hence the validity conditions for $g_1,g_3$, and $g_5$ follow from
Table~\ref{tab:known-pn-q3n}; in particular, $g_5$ is a member of the
Bierbrauer $f_5$ family~\cite[Thm.~4]{Bier10}.
Moreover, $s\nmid t$ excludes $t=2s$, the only value for which the two
exponents of $g_5$ are congruent modulo $3^{3s}-1$; hence $g_5$ is a
genuine binomial.
Semifields in the isotopy classes associated with $g_1,g_3$, and $g_5$
have nucleus and middle-nucleus orders
\begin{equation}\label{eq:three-class-nuclei}
  (3^s,3^s),
  \qquad
  (3,3),
  \qquad
  (3^\ell,3^\ell),
\end{equation}
respectively.  The first two pairs in \eqref{eq:three-class-nuclei} are
recorded, respectively, in
\cite[Sec.~3, p.~830]{AndreoliEtAl25} and
\cite[Sec.~3, p.~831]{AndreoliEtAl25}; see also
\cite{CoulterHenderson08}.  In the notation of
\cite[Sec.~4, Thm.~4.3, p.~18]{GologluKolsch26}, the function $g_5$ is
$F_{3^t,\omega}$ over $\F_{3^{3s}}$ and is
$\operatorname{GL}(3,3^s)$-equivalent to a function in the family of
their Theorem~1.9.  This equivalence induces a strong isotopy between
the associated presemifields~\cite[Sec.~2.7, p.~10]{GologluKolsch26}.
The automorphism $x\mapsto x^{3^t}$ of $\F_{3^s}$ has fixed field
$\F_{3^\ell}$, so
\cite[Appendix~A, Thm.~A.4, p.~23]{GologluKolsch26} gives the third pair in
\eqref{eq:three-class-nuclei}.  Since $1<\ell<s$, the three pairs in
\eqref{eq:three-class-nuclei} are distinct.
The orders of the nuclei are isotopy invariants
\cite[Sec.~2.2, p.~827]{AndreoliEtAl25}; hence the three associated
semifields are pairwise nonisotopic.  For planar DO
polynomials, CCZ equivalence is equivalent to strong isotopy
\cite[Sec.~2.3, p.~827]{AndreoliEtAl25}, and strong isotopy implies
isotopy.  Therefore the three inputs are pairwise CCZ-inequivalent.
Corollary~\ref{cor:ccz-inequivalent-associated-functions} then shows that
$\widetilde G_1,\widetilde G_3$, and $\widetilde G_5$ are pairwise
CCZ-inequivalent.
Finally, $s=5\ell$ and $t=\ell$ give $(s+t)/\ell=6$ and satisfy the other
conditions immediately.
\end{proof}

\section{Conclusion and further directions}
\label{sec:conclusion}

We have shown that every DO PN function $g$ over $\Fq$, where $q=3^n$ and
$n>1$ is odd, gives rise to a family of APN functions whose exact differential
and boomerang spectra can be determined.  Writing $\tau=g(1)$, the sign-switch
$G=g+\tau\epsilon$ is APN with $\delta_G=\beta_G=2$, and every function
\[
  \widetilde G_c(x):=G(x)+D_cg(x)=g(x+c)+\tau\epsilon(x),
  \qquad c\in\Fq,
\]
is APN.  Exactly $(q-3)/2$ choices of $c$ yield boomerang uniformity one, and
the remaining $(q+3)/2$ yield boomerang uniformity two.  This family has a
common exact differential spectrum, while its complete boomerang spectra fall
into three types in the former case and five in the latter.  We also proved
that, for functions over finite fields of odd characteristic with differential
uniformity at most two, boomerang uniformity zero implies perfect
nonlinearity; hence boomerang uniformity one is the least possible value for
an APN function in odd characteristic.

The common differential spectrum also separates the constructed functions
from every power function and every Ness--Helleseth-type binomial;
its CCZ-invariance therefore yields CCZ-inequivalence from both classes.
Furthermore, switch-rigidity shows that CCZ equivalence between sign-switches
of DO PN functions forces EA equivalence between their PN inputs.  Combining
this result with the different orders of the nuclei of the associated
presemifields yields three pairwise
CCZ-inequivalent APN functions with boomerang uniformity one for infinitely
many odd extension degrees, using the Gold, Ding--Yuan, and Bierbrauer
families.  The first degree obtained by this argument is $n=45$.

The present method relies essentially on the DO property, which makes every
derivative $D_cg$ affine.  A natural next case is therefore the
Coulter--Matthews PN family $x^{(3^k+1)/2}$, with $k$ odd and
$\gcd(n,k)=1$.  This family is generally non-DO, so the present
affine-perturbation argument does not apply directly.  It would be interesting
to determine the exact spectra of its sign-switches and to find a modified
perturbation that yields APN functions with boomerang uniformity one.  More
broadly, one may ask whether analogous switching and
perturbation methods can produce APN functions with boomerang uniformity one
beyond the ternary DO setting.

\newpage
\noindent\textbf{Acknowledgments:}
Soonhak Kwon was supported by Basic Science Research Program through the
National Research Foundation of Korea (NRF) funded by the Ministry of Education
(No.~RS-2019-NR040081). Namhun Koo was supported by Basic Science Research
Program through the National Research Foundation of Korea (NRF) funded by the
Ministry of Education (No.~RS-2026-25575984).


\begin{thebibliography}{99}
\setlength{\itemsep}{1pt}

\bibitem{AndreoliEtAl25}
S. Andreoli, L. Budaghyan, R. S. Coulter, A. Haukenes, N. Kaleyski, and E. Piccione,
\emph{On a classification of planar functions in characteristic three},
Cryptogr. Commun. \textbf{17} (2025), 823--853.
DOI: 10.1007/s12095-025-00781-y.

\bibitem{BartoliStanica26}
D. Bartoli and P. St\u{a}nic\u{a},
\emph{Non-existence of infinite APN families from patched monomials in odd characteristic},
Cryptogr. Commun. (2026), published online.
DOI: 10.1007/s12095-026-00880-4.

\bibitem{Bier10}
J. Bierbrauer,
\emph{New semifields, PN and APN functions},
Des. Codes Cryptogr. \textbf{54} (2010), 189--200. DOI: 10.1007/s10623-009-9318-7.

\bibitem{BudaghyanPal25}
L. Budaghyan and M. Pal,
\emph{Arithmetization-oriented APN permutations},
Des. Codes Cryptogr. \textbf{93} (2025), no.~4, 1067--1088.
DOI: 10.1007/s10623-024-01487-7.

\bibitem{CCZ98}
C. Carlet, P. Charpin, and V. Zinoviev,
\emph{Codes, bent functions and permutations suitable for DES-like cryptosystems},
Des. Codes Cryptogr. \textbf{15} (1998), 125--156.
DOI: 10.1023/A:1008344232130.

\bibitem{CHNC13}
S.-T. Choi, S. Hong, J.-S. No, and H. Chung,
\emph{Differential spectrum of some power functions in odd prime characteristic},
Finite Fields Appl. \textbf{21} (2013), 11--29.

\bibitem{CidEtAl18}
C. Cid, T. Huang, T. Peyrin, Y. Sasaki, and L. Song,
\emph{Boomerang connectivity table: A new cryptanalysis tool},
in Advances in Cryptology--EUROCRYPT 2018, Lecture Notes in Computer Science
\textbf{10821}, Springer, 2018, 683--714.
DOI: 10.1007/978-3-319-78375-8\_22.

\bibitem{CoulterHenderson08}
R. S. Coulter and M. Henderson,
\emph{Commutative presemifields and semifields},
Adv. Math. \textbf{217} (2008), no.~1, 282--304.
DOI: 10.1016/j.aim.2007.07.007.

\bibitem{CM97}
R. S. Coulter and R. W. Matthews,
\emph{Planar functions and planes of Lenz--Barlotti class II},
Des. Codes Cryptogr. \textbf{10} (1997), no. 2, 167--184.

\bibitem{CM11}
R. S. Coulter and R. W. Matthews,
\emph{On the number of distinct values of a class of functions over a finite field},
Finite Fields Appl. \textbf{17} (2011), no. 3, 220--224. DOI: 10.1016/j.ffa.2010.12.002.

\bibitem{DO68}
P. Dembowski and T. G. Ostrom,
\emph{Planes of order $n$ with collineation groups of order $n^2$},
Math. Z. \textbf{103} (1968), no. 3, 239--258.

\bibitem{Dempwolff18}
U. Dempwolff,
\emph{CCZ equivalence of power functions},
Des. Codes Cryptogr. \textbf{86} (2018), 665--692.

\bibitem{Dempwolff22}
U. Dempwolff,
\emph{Correction to: CCZ equivalence of power functions},
Des. Codes Cryptogr. \textbf{90} (2022), 473--475.
DOI: 10.1007/s10623-021-00979-0.

\bibitem{DY06}
C. Ding and J. Yuan,
\emph{A family of skew Hadamard difference sets},
J. Combin. Theory Ser. A \textbf{113} (2006), 1526--1535.

\bibitem{DMM03}
H. Dobbertin, D. Mills, E. N. M\"uller, A. Pott, and W. Willems,
\emph{APN functions in odd characteristic},
Discrete Math. \textbf{267} (2003), 95--112.

\bibitem{FL07}
K. Feng and J. Luo,
\emph{Value distributions of exponential sums from perfect nonlinear functions and their applications},
IEEE Trans. Inf. Theory \textbf{53} (2007), no. 9, 3035--3041. DOI: 10.1109/TIT.2007.903153.

\bibitem{GologluKolsch26}
F. G\"olo\u{g}lu and L. K\"olsch,
\emph{Commutative semifields from bijections of the Desarguesian plane},
J. Lond. Math. Soc. \textbf{114} (2026), no.~1, Art.~e70635.
DOI: 10.1112/jlms.70635.

\bibitem{HuEtAl23}
Z. Hu, N. Li, L. Xu, X. Zeng, and X. Tang,
\emph{The differential spectrum and boomerang spectrum of a class of locally-APN functions},
Des. Codes Cryptogr. \textbf{91} (2023), no.~5, 1695--1711.
DOI: 10.1007/s10623-022-01161-w.

\bibitem{KK26}
N. Koo and S. Kwon,
\emph{On differential and boomerang properties of a class of binomials over finite fields of odd characteristic},
IEEE Trans. Inf. Theory \textbf{72} (2026), no.~3, 1928--1942.
DOI: 10.1109/TIT.2026.3657603.

\bibitem{KKKB26b}
N. Koo, S. Kwon, M. Ko, and B. Kim,
\emph{Locally-APN binomials with low boomerang uniformity in odd characteristic},
arXiv:2512.17603v2, 2026.

\bibitem{KKKB26}
N. Koo, S. Kwon, M. Ko, and B. Kim,
\emph{On APN exponents and the differential and boomerang properties of binomials in characteristic $3$},
arXiv:2605.23224, 2026.

\bibitem{LangeWinterhof02}
T. Lange and A. Winterhof,
\emph{Incomplete character sums over finite fields and their application to the
interpolation of the discrete logarithm by Boolean functions},
Acta Arith. \textbf{101} (2002), no.~3, 223--229.
DOI: 10.4064/aa101-3-3.

\bibitem{Leducq12}
E. Leducq,
\emph{New families of APN functions in characteristic $3$ or $5$},
in Arithmetic, Geometry, Cryptography and Coding Theory,
Contemp. Math. \textbf{574}, Amer. Math. Soc., 2012, 115--123.

\bibitem{LiQuSunLi19}
K. Li, L. Qu, B. Sun, and C. Li,
\emph{New results about the boomerang uniformity of permutation polynomials},
IEEE Trans. Inf. Theory \textbf{65} (2019), no.~11, 7542--7553.
DOI: 10.1109/TIT.2019.2918531.

\bibitem{LyuWangZheng24}
C. Lyu, X. Wang, and D. Zheng,
\emph{A further study on the Ness--Helleseth function},
Finite Fields Appl. \textbf{98} (2024), Art.~102453.
DOI: 10.1016/j.ffa.2024.102453.

\bibitem{MW25}
S. Mesnager and H. Wu,
\emph{The differential and boomerang properties of a class of binomials},
IEEE Trans. Inf. Theory \textbf{71} (2025), no.~6, 4854--4871.
DOI: 10.1109/TIT.2025.3550851.

\bibitem{NH07}
G. J. Ness and T. Helleseth,
\emph{A new family of ternary almost perfect nonlinear mappings},
IEEE Trans. Inf. Theory \textbf{53} (2007), no.~7, 2581--2586.
DOI: 10.1109/TIT.2007.899508.

\bibitem{Nyberg94}
K. Nyberg,
\emph{Differentially uniform mappings for cryptography},
in Advances in Cryptology--EUROCRYPT '93, Lecture Notes in Computer Science
\textbf{765}, Springer, 1994, 55--64.
DOI: 10.1007/3-540-48285-7\_6.

\bibitem{PalStanica25}
M. Pal and P. St\u{a}nic\u{a},
\emph{A connection between the boomerang uniformity and the extended differential in odd characteristic and applications},
Adv. Math. Commun. \textbf{19} (2025), no.~5, 1382--1403.
DOI: 10.3934/amc.2024059.

\bibitem{XiaEtAl24}
Y. Xia, F. Bao, S. Chen, C. Li, and T. Helleseth,
\emph{More differential properties of the Ness--Helleseth function},
IEEE Trans. Inf. Theory \textbf{70} (2024), no.~8, 6076--6090.
DOI: 10.1109/TIT.2024.3408882.

\bibitem{XZLH20}
Y. Xia, X. Zhang, C. Li, and T. Helleseth,
\emph{The differential spectrum of a ternary power mapping},
Finite Fields Appl. \textbf{64} (2020), 101660.

\bibitem{XuCaoXu16}
G. Xu, X. Cao, and S. Xu,
\emph{Constructing new APN functions and bent functions over finite fields of odd characteristic via the switching method},
Cryptogr. Commun. \textbf{8} (2016), 155--171.
DOI: 10.1007/s12095-015-0145-6.

\bibitem{YMT24}
H. Yan, S. Mesnager, and X. Tan,
\emph{On a class of APN power functions over odd characteristic finite fields: their differential spectrum and $c$-differential properties},
Discrete Math. \textbf{347} (2024), 113881.

\bibitem{YXLHL22}
H. Yan, Y. Xia, C. Li, T. Helleseth, M. Xiong, and J. Luo,
\emph{The differential spectrum of the power mapping $x^{p^n-3}$},
IEEE Trans. Inf. Theory \textbf{68} (2022), no. 8, 5535--5547.

\bibitem{ZengHuYangJiang07}
X. Zeng, L. Hu, Y. Yang, and W. Jiang,
\emph{On the inequivalence of Ness--Helleseth APN functions},
IACR Cryptology ePrint Archive, Report 2007/379, 2007.

\bibitem{ZhaHu13}
Z. Zha and L. Hu,
\emph{Constructing new APN functions from known PN functions},
Int. J. Found. Comput. Sci. \textbf{24} (2013), no.~8, 1209--1219.
DOI: 10.1142/S0129054113500299.

\bibitem{ZKW09}
Z. Zha, G. M. Kyureghyan, and X. Wang,
\emph{Perfect nonlinear binomials and their semifields},
Finite Fields Appl. \textbf{15} (2009), no. 2, 125--133. DOI: 10.1016/j.ffa.2008.09.002.

\bibitem{ZW10}
Z. Zha and X. Wang,
\emph{Power functions with low uniformity on odd characteristic finite fields},
Sci. China Math. \textbf{53} (2010), 1931--1940.

\end{thebibliography}
\end{document}